\documentclass[%
reprint,
superscriptaddress,
 amsmath,amssymb,
 aps,
 pra,
]{revtex4-2}

\usepackage{graphicx}
\usepackage{dcolumn}
\usepackage{bm}
\usepackage{float}
\usepackage{comment}
\usepackage{hyperref}

\usepackage[T1]{fontenc}

\usepackage[caption=false]{subfig}

\usepackage{amsmath}
\usepackage{amssymb}
\usepackage{mathtools}

\usepackage{amsthm}
\newtheorem{theorem}{Theorem}[section]

\newtheorem{proposition}[theorem]{Proposition}
\newtheorem*{remark}{Remark}

\makeatletter
\newsavebox{\@brx}
\newcommand{\llangle}[1][]{\savebox{\@brx}{\(\m@th{#1\langle}\)}%
  \mathopen{\copy\@brx\kern-0.5\wd\@brx\usebox{\@brx}}}
\newcommand{\rrangle}[1][]{\savebox{\@brx}{\(\m@th{#1\rangle}\)}%
  \mathclose{\copy\@brx\kern-0.5\wd\@brx\usebox{\@brx}}}
\makeatother

\begin{document}

\preprint{APS/123-QED}

\title{Entanglement Purification in Quantum Networks:\\ Guaranteed Improvement and Optimal Time}

\author{Allen Zang}
\affiliation{Pritzker School of Molecular Engineering, University of Chicago, Chicago, IL, USA}

\author{Xin-An Chen}
\affiliation{Department of Electrical and Computer Engineering, University of Illinois Urbana-Champaign, Urbana, IL, USA}

\author{Eric Chitambar}
\affiliation{Department of Electrical and Computer Engineering, University of Illinois Urbana-Champaign, Urbana, IL, USA}

\author{Martin Suchara}
\affiliation{Microsoft Azure Quantum, Microsoft Corporation, Redmond, WA, USA}

\author{Tian Zhong}
\affiliation{Pritzker School of Molecular Engineering, University of Chicago, Chicago, IL, USA}

\date{\today}

\begin{abstract}
    While the concept of entanglement purification protocols (EPPs) is straightforward, the integration of EPPs in network architectures requires careful performance evaluations and optimizations that take into account realistic conditions and imperfections, especially probabilistic entanglement generation and quantum memory decoherence. 
    It is important to understand what is guaranteed to be improved from successful EPP with arbitrary non-identical input, which determines whether we want to perform the EPP at all. When successful EPP can offer improvement, the time to perform the EPP should also be optimized to maximize the improvement. 
    In this work, we study the guaranteed improvement and optimal time for the CNOT-based recurrence EPP, previously shown to be optimal in various scenarios. 
    We firstly prove guaranteed improvement for multiple figures of merit, including fidelity and several entanglement measures when compared to practical baselines as functions of input states. However, it is noteworthy that the guaranteed improvement we prove does not imply the universality of the EPP as introduced in \href{https://arxiv.org/abs/2407.21760}{arXiv:2407.21760}. Then we prove robust, parameter-independent optimal time for typical error models and figures of merit. We further explore memory decoherence described by continuous-time Pauli channels, and demonstrate the phenomenon of optimal time transition when the memory decoherence error pattern changes.
    Our work deepens the understanding of EPP performance in realistic scenarios and offers insights into optimizing quantum networks that integrate EPPs.
\end{abstract}

\maketitle

\section{Introduction}\label{sec:intro}
Quantum networks~\cite{kimble2008quantum,wehner2018quantum} promise to drive new scientific and technological advances in distributed quantum information processing, such as quantum key distribution~\cite{bennett2014quantum}, distributed quantum computing~\cite{gottesman1999demonstrating, jiang2007distributed, monroe2014large,cacciapuoti2019quantum,cuomo2020towards,barral2024review}, and distributed quantum sensing~\cite{komar2014quantum,proctor2018multiparameter,zhang2021distributed}. In the near term, the first generation (1G) quantum repeater~\cite{briegel1998quantum,munro2015inside,muralidharan2016optimal,azuma2023quantum} is the most suitable architecture for realizing quantum networks, where heralded entanglement generation~\cite{moehring2007entanglement,ritter2012elementary,hofmann2012heralded,bernien2013heralded}, entanglement swapping~\cite{zukowski1993event,pan1998experimental}, and entanglement purification protocols (EPPs)~\cite{bennett1996purification,deutsch1996quantum,dur2007entanglement} are utilized to distribute high-fidelity entanglement between quantum network nodes. Recent experimental advances allowed laboratory-scale~\cite{pompili2021realization,hermans2022qubit} and even metropolitan-scale~\cite{knaut2024entanglement,liu2024creation,stolk2024metropolitan} quantum network demonstrations. 
Despite these experimental milestones, the achievable fidelity of entanglement distribution still requires significant improvements to enable practical applications. Therefore, it is critically important to integrate and optimize the operation of EPPs in practical quantum networks beyond conceptual demonstrations~\cite{pan2001entanglement,pan2003experimental,hu2021long,ecker2021experimental,reichle2006experimental,kalb2017entanglement,yan2022entanglement}. 
While crucial, the optimization of the EPP policy in quantum networks is also extremely difficult due to the size of the state space~\cite{khatri2021towards,khatri2021policies,khatri2022design}. Specifically, there is a temporal degree of freedom, i.e. the EPP can be performed at different times between the time when input states become available and the time when entangled states need to be consumed, such as in quantum repeaters with a buffer time~\cite{santra2019quantum,zang2023entanglement} or with continuous entanglement generation~\cite{chakraborty2019distributed,kolar2022adaptive,inesta2023performance,ghaderibaneh2022pre,zhan2025design}, where there is almost always a time period between successful entanglement generation and utilization of the entanglement. Also, the probabilistic nature of entanglement generation together with quantum memory decoherence introduces time-dependent and non-identical input entangled states. 

Different from prior studies that evaluated the overall performance of specific families of quantum repeater network architectures with EPP~\cite{dur1999quantum,ladd2006hybrid,hartmann2007role,razavi2009quantum,razavi2009physical,bratzik2013quantum,zang2023entanglement,victora2023entanglement,zang2024quantum,mantri2024comparing,haldar2025reducing}, we aim at further understanding the properties of EPPs in practical quantum network scenarios to guide the optimization of EPP operation. This has remained largely unexplored, but it is practically important as the study of universality for EPPs~\cite{zang2024no} has revealed fundamental limits in EPP performance. Our two central questions are:
\begin{enumerate}
    \item Can the EPP guarantee any type of improvement conditioned on success, when the input states are arbitrary and thus generally non-identical (for instance, the states can have different fidelities or differ in other figures of merit)?
    \item When entangled states are generated successively due to probabilistic entanglement generation, and quantum memories also decohere over time, what is the best scheduling strategy to perform the EPP?
\end{enumerate}
Both questions emerge naturally from quantum network scenarios with probabilistic entanglement generation and noisy quantum memories. Moreover, they are practically important because the answer to the first question determines whether we want to perform the EPP at all in each given scenario, while the answer to the second question offers us the insight on how to improve the quality of entanglement distribution through the timing of the EPP. For the first question, we perform extensive analytical studies of different forms of noisy input states and figures of merit, and prove various practical quantities that are guaranteed to be improved conditioned on success. It is noteworthy that the baselines for guaranteed improvement do not necessarily involve only the better input state, so that the no-go theorems of universal EPPs~\cite{zang2024no} are not violated. For the second question, we combine analytics with numerics to prove parameter-independent optimal time for performing the EPP, and demonstrate the intriguing phenomenon of optimal time transition when the decoherence model varies.

Among the various types of EPPs, we focus on the canonical BBPSSW~\cite{bennett1996purification}/DEJMPS~\cite{deutsch1996quantum} (recurrence) protocol based on bilocal CNOT because it is arguably the most experiment friendly EPP~\cite{pan2001entanglement,pan2003experimental,hu2021long,ecker2021experimental,reichle2006experimental,kalb2017entanglement,yan2022entanglement}. In addition, recent analytical and computational studies~\cite{rozpkedek2018optimizing,preti2022optimal,jansen2022enumerating} imply its optimality for identical input states of different forms, and there is also evidence of its optimality with non-identical input states~\cite{chen_in_prep}. We note that if input states are known, additional ad hoc optimizations of EPPs~\cite{krastanov2019optimized} can improve the output. However, in practical quantum networks the input states are not static but dynamic, which means that the exact characterization of input states for each EPP shot is impractical. This further justifies our choice of focusing on a fixed EPP.

The paper is organized as follows. In Sec.~\ref{sec:prelim}, we review the necessary background and the chosen figures of merit that will be used for performance evaluations. In Sec.~\ref{sec:improv}, we demonstrate guaranteed improvement from the EPP by comparing the successful output of EPP given non-identical input states, with varying baselines constructed from the input states, using the figures of merit mentioned in Sec.~\ref{sec:prelim}. We further study the optimal time to perform the EPP under a prototypical, realistic quantum network scenario with probabilistic entanglement generation and quantum memory decoherence in Sec.~\ref{sec:opt_time}. Conclusion and discussion are in Sec.~\ref{sec:discussion} and proofs are in the appendices.

\section{Preliminaries}\label{sec:prelim}

\subsection{Recurrence entanglement purification protocol}
The EPP we investigate takes two noisy two-qubit entangled states between Alice and Bob, $\rho_1$ and $\rho_2$, as input. Alice and Bob apply the CNOT gate to the two qubits each of them holds, treating one qubit from $\rho_1$ as the control and the other from $\rho_2$ as the target. Then both parties measure the target qubit in the computational basis and communicate the classical measurement result to each other. The purification process is successful if both measurement results have equal parity, and unsuccessful otherwise. If unsuccessful, the unmeasured qubit pair is also discarded. The circuit diagram of this EPP is illustrated in Fig.~\ref{fig:recurrence_EPP}. 
\begin{figure}[t]
    \centering
    \includegraphics[width=0.5\columnwidth]{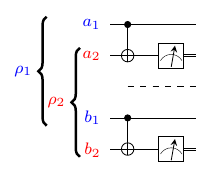}
    \caption{Circuit of the recurrence EPP.}
    \label{fig:recurrence_EPP}
\end{figure}

We will restrict ourselves to Bell diagonal states (BDS) whose density matrices are diagonal in the Bell basis, as any two-qubit states can be converted to a BDS with fidelity unchanged via Pauli twirling~\cite{dur2005standard,emerson2007symmetrized,dankert2009exact}. By tracking how the input states are transformed throughout the circuit,  the output fidelity conditioned on success and the success probability can be derived~\cite{dur2007entanglement} for BDS. For general BDS in the form of $\rho_\mathrm{BDS}(\vec{\lambda})=\lambda_1\Phi^+ + \lambda_2\Phi^- + \lambda_3\Psi^+ + \lambda_4\Psi^-$, $\lambda_1+\lambda_2+\lambda_3+\lambda_4=1$ as input, one characterized by $\vec{\lambda}=(\lambda_1,\lambda_2,\lambda_3,\lambda_4)$ and the other by $\vec{\lambda}'=(\lambda_1',\lambda_2',\lambda_3',\lambda_4')$, the successful output state is described by the following transformed parameters
\begin{subequations}\label{eqn:bds_elems_dejmps}
\begin{eqnarray}
&&\lambda_\mathrm{1,succ} = \frac{\lambda_1\lambda_1' + \lambda_2\lambda_2'}{p_\mathrm{BDS}} = F_\mathrm{BDS},\\
&& \lambda_\mathrm{2,succ} = \frac{\lambda_1\lambda_2' + \lambda_2\lambda_1'}{p_\mathrm{BDS}},\\
&& \lambda_\mathrm{3,succ} = \frac{\lambda_3\lambda_3' + \lambda_4\lambda_4'}{p_\mathrm{BDS}},\\
&& \lambda_\mathrm{4,succ} = \frac{\lambda_3\lambda_4' + \lambda_4\lambda_3'}{p_\mathrm{BDS}},
\end{eqnarray}
\end{subequations}
where the first diagonal element of the output state in the Bell basis $\lambda_\mathrm{1,succ}$ is the output fidelity $F_\mathrm{BDS}$, and the denominator is the success probability
\begin{equation}\label{eqn:bds_succ_prob}
    p_\mathrm{BDS} = (\lambda_1 + \lambda_2)(\lambda_1' + \lambda_2') + (\lambda_3 + \lambda_4)(\lambda_3' + \lambda_4').
\end{equation}

\subsection{Figures of merit}\label{sec:fom}
Here we describe the figures of merit we use throughout the work. We focus on metrics that quantify entanglement quality. Quantum communication protocols such as gate and qubit teleportation require entangled states be in specific Bell state forms. The first and perhaps the most important figure of merit for EPP performance is the fidelity $F$ of successful output state $\rho$ w.r.t. a specific \textit{pure} Bell state $\psi_\mathrm{Bell}$, i.e. $F=\mathrm{Tr}(\rho \psi_\mathrm{Bell})$, where $\psi_\mathrm{Bell}$ denotes the density matrix of a pure Bell state (one of $\Phi^+,\Phi^-,\Psi^+,\Psi^-$). Other metrics we consider are entanglement measures (monotones), including concurrence, (logarithmic) negativity, and distillable entanglement. In the following we provide their explicit formulae for BDS.

\subsubsection{Concurrence}
Concurrence~\cite{hill1997entanglement,wootters1998entanglement} for a two-qubit state $\rho$, $\mathcal{C}(\rho)$, is defined as $\mathcal{C}(\rho) = \max\{0, e_1 - e_2 - e_3 - e_4\}$, where $e_i$ are eigenvalues in decreasing order of operator $R(\rho) = \sqrt{\sqrt{\rho}\tilde{\rho}\sqrt{\rho}}$ with $\tilde{\rho}= (\sigma_y\otimes\sigma_y)\rho^*(\sigma_y\otimes\sigma_y)$ being the spin-flipped state of $\rho$, and $\rho^*$ being complex conjugate of $\rho$ written in $\sigma_z$ basis. For BDS we have the following expression of concurrence
\begin{equation}\label{eqn:concurrence}
    \mathcal{C}(\rho_\mathrm{BDS}) = \max\{0, \lambda_1-\lambda_2-\lambda_3-\lambda_4\} = 2\lambda_1 - 1,
\end{equation}
which is still linear in $F$ as $\lambda_1$ equals to fidelity $F$, thus qualitatively the same as the fidelity for BDS.

\subsubsection{Logarithmic negativity}\label{sec:negativity}
(Logarithmic) negativity~\cite{plenio2005logarithmic} is derived from the positive partial transpose (PPT) criterion~\cite{peres1996separability,horodecki2001separability}. The negativity $\mathcal{N}$ and logarithmic negativity $E_N$ are defined as $\mathcal{N}(\rho_{AB}) = \left\vert\sum_{e_i<0}e_i\right\vert = \sum_i\frac{\vert e_i\vert-e_i}{2}$ and $E_N(\rho_{AB}) = \log_2\left[2\mathcal{N}(\rho_{AB}) + 1\right]$, respectively, where $\rho_{AB}^\mathrm{T_A}$ is the partial transpose of the bi-partite state $\rho_{AB}$ w.r.t. the subsystem $A$, and $e_i$ are eigenvalues of $\rho_{AB}^\mathrm{T_A}$. Then for all BDS we have its negativity
\begin{equation}
    \mathcal{N}(\rho_\mathrm{BDS}) = \frac{\lambda_1-\lambda_2-\lambda_3-\lambda_4}{2} = \frac{2\lambda_1 - 1}{2},
\end{equation}
which is linear in fidelity $F$, and only differs from concurrence by a factor of 1/2, thus these two are equivalent. Correspondingly, the logarithmic negativity is
\begin{equation}\label{eqn:log_neg}
    E_N(\rho_\mathrm{BDS}) = \log_2(2F - 1 + 1) = 1 + \log_2(F),
\end{equation}
which equals 1 when $F=1$, i.e. pure Bell state $\Psi_{AB}^+$, as expected. 

\subsubsection{Distillable entanglement}
Now we consider an operational entanglement measure, the distillable entanglement $E_D(\rho)$ which quantifies the rate of maximally entangled state that can be distilled from the ensemble of state $\rho$. In general, distillable entanglement is difficult to compute, but there exist lower and upper bounds. For lower bound we consider the (one-way) coherent information~\cite{schumacher1996quantum}, and for upper bound we consider the Rains' bound~\cite{rains1999bound,rains2001semidefinite}. For BDS with fidelity $F$, we can calculate the Rains' bound according to Theorem 8 in~\cite{rains1999bound} as $R(F) = 1 - H_b(F)$. For two-qubit state $\rho_{AB}$, the coherent information is defined as $I_{A\rangle B} = -H(A\vert B) = H(\rho_B) - H(\rho_{AB})$, where $\rho_B = \mathrm{tr}_A(\rho_{AB})$, and $H(\rho) = -\mathrm{tr}[\rho\log(\rho)]$ is the von Neumann entropy. Then for all BDS the coherent information is
\begin{equation}
    I_{A\rangle B}^\mathrm{BDS} = 1 + \sum_{i=1}^4\lambda_i\log_2 \lambda_i.
\end{equation}

Now we consider two examples of BDS. For rank-2 BDS with fidelity $F$ we have 
\begin{equation}
    I_{A\rangle B}^{(2)} = 1 - H_b(F),
\end{equation}
where $H_b(p) = -p\log_2 p - (1-p)\log_2(1-p)$ is the binary entropy function, and the superscript denotes rank-2. Similarly for the Werner state we have
\begin{equation}
    I_{A\rangle B}^{W} = 1 + \left[F\log_2 F + (1-F)\log_2\frac{1-F}{3}\right].
\end{equation}
Therefore, the distillable entanglement for rank-2 BDS is exactly $E_D(F) = 1 - H_b(F)$, while for the Werner state the upper and lower bounds are not tight.

\subsection{Dynamics of BDS stored in noisy quantum memory}
To study the optimal time of entanglement purification, it is essential to know how the entangled state stored in the quantum memories evolves over time under memory decoherence. As we have mentioned, this work focuses on BDS whose density matrix has at most four non-zero elements. We are able to derive analytical formulae of how the four diagonal density matrix elements of a noisy EPR pair evolve under two independent single-qubit continuous-time Pauli channels that describe the decoherence process of two remote quantum memories which are used to store one EPR pair over time. There are two main factors for this error model, namely the \textit{error pattern} which represents the ratio between three Pauli error components, and the \textit{decoherence rate} which determines the time scale of fidelity decay from the decoherence channel, and can be understood as the probability of having some error within a unit time. The explicit formulae and detailed derivations are presented in Appendix~\ref{sec:bds_dynamics}. Note that the independence of qubit error channels is justified by the fact that qubits (quantum memories) on different network nodes are far away from each other, thus there is effectively zero physical interaction that can generate correlation.

\section{Guaranteed Improvement}\label{sec:improv}
To explore what is guaranteed to be improved from successful entanglement purification, we consider all possible combinations of two input states with fidelity above 1/2 (to make sure that there is entanglement), for rank-2 BDS, Werner state, and general BDS. We assume that the two input states have identical parametrization and only differ in fidelities.

\subsection{Rank-2 BDS}
Without loss of generality, for the specific protocol chosen for this work, we choose rank-2 BDS which corresponds to $|\phi^+\rangle$ undergoing bit-flip channel (Pauli $X$ error), i.e. $\rho_\mathrm{2-BDS}(F)=F\Phi^+ + (1-F)\Psi^+$. We are interested in the output fidelity w.r.t. the fidelities of the two input states. For this class of input states, as long as the input fidelities are above 1/2, the output fidelity conditioned on success is guaranteed to be no lower than the higher input fidelity.
\begin{proposition}\label{thm:rank2bds_fidincr}
    Let $F^{(2)}(F_1,F_2)$ be the successful output fidelity of the recurrence EPP. Given two rank-2 BDS input states $\rho_\mathrm{2-BDS}(F_1)$ and $\rho_\mathrm{2-BDS}(F_2)$ with $F_1,F_2\geq 1/2$, we have $F^{(2)}(F_1,F_2)\geq\max\{F_1,F_2\}$.
\end{proposition}
This result suggests that the 2-to-1 recurrence EPP is \textit{universal}~\cite{zang2024no} for rank-2 BDS with fidelity threshold 1/2. This fact can be well understood by considering the error detection mechanism of this specific EPP, which can detect a Pauli $X$ error on the Bell pair. Thus post-selection allows the improvement of fidelity.

Recall other figures of merit other than fidelity, such as concurrence and negativity which are linear in BDS fidelity, and also those nonlinear in fidelity such as logarithmic negativity and distillable entanglement. All the aforementioned entanglement measures monotonically increase with increasing fidelity, and thus from Prop.~\ref{thm:rank2bds_fidincr} we conclude that the recurrence EPP guarantees the successfully purified state improved concurrence, (logarithmic) negativity, and distillable entanglement in comparison to both input states.

\subsection{Werner state}
Rank-2 BDS represent Bell states under completely biased noise with only one Pauli component. Now we switch to Werner state which corresponds to the unbiased depolarizing channel. First of all, it can be explicitly verified that with Werner states as input of the recurrence EPP, the output fidelity upon success is not guaranteed to be at least as high as the higher input fidelity for an arbitrary input fidelity pair $(F_1,F_2)\in[1/2,1]\times[1/2,1]$. For instance for two two-qubit Werner states with fidelities 0.9 and 0.95, the successful output fidelity is approximately $0.946<0.95$. This means that the recurrence EPP is \textit{not universal}~\cite{zang2024no} for Werner states, or general BDS. This is because the specific purification circuit can only detect a Pauli $X$ error but will fail for a Pauli $Z$ error, while a Werner state is a Bell state subject to not only $X$ (and $Y\simeq XZ$) error but also $Z$ error with equal strength. This is exactly why entanglement pumping~\cite{dur2003entanglement} cannot arbitrarily improve entanglement fidelity. 

Having made this observation, we ask if for Werner state inputs there is any quantity that is guaranteed to be improved. In fact, we can prove that the output fidelity is guaranteed to be higher than at most the average input fidelity. Specifically, we consider a parametrized baseline which is a convex combination of the higher input fidelity and the lower input fidelity: $F_\mathrm{baseline}(\lambda,F_1,F_2)=\lambda\max\{F_1,F_2\}+(1-\lambda)\min\{F_1,F_2\}$, and average input fidelity is equal to $F_\mathrm{baseline}(1/2,F_1,F_2)$.
\begin{proposition}\label{thm:werner_fidincr}
    Let $F^{(W)}(F_1,F_2)$ be the successful output fidelity of the recurrence EPP. Given two input Werner states $\rho_\mathrm{W}(F_1)$ and $\rho_\mathrm{W}(F_2)$ with $F_1,F_2\geq 1/2$, we have $F^{(W)}(F_1,F_2)\geq F_\mathrm{baseline}(\lambda,F_1,F_2)$ for $\lambda\in[0,1/2]$, while for $\lambda\in(1/2,1]$ there always exist $F_1,F_2\geq 1/2$ s.t. $F^{(W)}(F_1,F_2)< F_\mathrm{baseline}(\lambda,F_1,F_2)$.
\end{proposition}
The best baseline $F_\mathrm{baseline}(1/2,F_1,F_2)$ equals to the fidelity of the average input state (AIS) $\overline{\rho^\mathrm{in}}=[\rho_\mathrm{W}(F_1)+\rho_\mathrm{W}(F_2)]/2$, which describes the ensemble if we randomly pick $\rho_\mathrm{W}(F_1)$ or $\rho_\mathrm{W}(F_2)$ with equal probability. This corresponds to the scenario where we have no information about the input state and thus the only strategy is to pick randomly with no bias. We can thus operationally understand the result as that, with recurrence purification protocol, the performance of quantum information processing using the successful output is guaranteed to be better than using the input states directly, no matter what the input states are, because any state can be twirled into Werner form. In fact, we can further consider $\lambda>1/2$ for $F_\mathrm{baseline}(\lambda,F_1,F_2)$, and we discover that although there is no guaranteed improvement with respect to $F_\mathrm{baseline}(\lambda,F_1,F_2)$ on the entire $(F_1,F_2)\in[1/2,1]\times[1/2,1]$ plane, we still have guaranteed improvement for higher fidelity threshold as long as $\lambda<2/3$. The details can be found in Appendix~\ref{sec:proof_improv}.

Combining the above result and the monotonicity with BDS fidelity of the entanglement measures considered in this work, the concurrence (logarithmic) negativity, and distillable entanglement (upper and lower bound) of the output state are guaranteed to be improved w.r.t., at most, those of the AIS conditioned on success for Werner states. Note that for the distillable entanglement lower bound we assume that the output state of successful entanglement purification has also been twirled to maintain the Werner form, as in the original BBPSSW paper~\cite{bennett1996purification}.

As the upper and lower bounds of distillable entanglement considered in this work generally do not collapse to each other, the above statement does not guarantee the improvement of the exact distillable entanglement. However, it is possible to examine the improvement of exact distillable entanglement given explicit upper and lower bounds, according to the straightforward argument that if the lower bound of the output quantity is higher than the upper bound of the input quantity then the output is surely higher than the input. Indeed, we are able to analytically determine fidelity threshold of the input Werner state for guaranteed distillable entanglement improvement w.r.t. distillable entanglement of AIS as follows. We note that for this specific proposition we consider the twirled output state which has lower distillable entanglement than the untwirled one, and the guaranteed improvement in the case where the output state is not twirled into the Werner form can be studied similarly.
\begin{proposition}\label{thm:werner_de_improv}
    The distillable entanglement of the successful output state (after twriling) from the recurrence EPP with two input Werner states $\rho_\mathrm{W}(F_1)$ and $\rho_\mathrm{W}(F_2)$, is guaranteed to be higher than the distillable entanglement of the AIS $\overline{\rho^\mathrm{in}}=[\rho_\mathrm{W}(F_1)+\rho_\mathrm{W}(F_2)]/2$, if $F_1,F_2\geq F_\mathrm{th}$ with $F_\mathrm{th}\approx0.939$ being the fidelity threshold.
\end{proposition}

Since any two-qubit states can be twirled into the Werner form without changing the fidelity, the guaranteed improvement for Werner states proved above stands as the lower bound of the recurrence EPP's capability with potentially non-identical input states. In the next section we consider the case where arbitrary BDS's are input to the EPP while twirling is not performed on the input state.

\subsection{General BDS}
Now we would like to examine if the above improvement is still guaranteed for other forms of BDS as input for the recurrence purification protocol. As usual, we consider general Bell diagonal states in the form of $\rho_\mathrm{BDS}(\vec{\lambda}) = \lambda_1\Phi^+ + \lambda_2\Phi^- + \lambda_3\Psi^+ + \lambda_4\Psi^-$ where $\lambda_1$ is equal to the fidelity $F$. We use $\lambda_1$ and $F$ interchangeably. We will consider input states with potentially different fidelities in identical form, i.e. with identical ratio between the last three components $\lambda_2:\lambda_3:\lambda_4$ corresponding to an identical error model. According to Eqn.~\ref{eqn:bds_elems_dejmps} and Eqn.~\ref{eqn:bds_succ_prob}, we have that for BDS with $\lambda_2=a(1-F)$ and $\lambda_3+\lambda_4 = (1-a)(1-F)$ where $a\in[0,1]$, the successful output fidelity and the success probability are only functions of $a$ and input fidelities $F_1$ and $F_2$.

We have seen that for both rank-2 BDS with $a=0$ and Werner states with $a=1/3$ fidelity improvement w.r.t. the fidelity of the AIS (equal to the average input fidelity) is guaranteed. We are thus curious if such guaranteed improvement still holds for other BDS with different $a$. We derive the following results.
\begin{proposition}\label{thm:general_bds_guarantee}
    Conditioned on successful recurrence EPP with two Bell diagonal input states $\rho_\mathrm{BDS}(\vec{\lambda}) = \lambda_1\Phi^+ + \lambda_2\Phi^- + \lambda_3\Psi^+ + \lambda_4\Psi^-$ and $\rho_\mathrm{BDS}(\vec{\lambda}') = \lambda_1'\Phi^+ + \lambda_2'\Phi^- + \lambda_3'\Psi^+ + \lambda_4'\Psi^-$ which satisfy $\lambda_2:\lambda_3:\lambda_4 = \lambda_2':\lambda_3':\lambda_4'$ and $\lambda_2/(1-\lambda_1)=\lambda_2'/(1-\lambda_1')=a$, the following statements hold:
    \begin{enumerate}
        \item The output fidelity $F_\mathrm{BDS}$ is always greater than the fidelity of AIS $(F_1+F_2)/2$ for all $F_1,F_2\in[1/2,1]$, if $a\leq 1/3$.
        \item For $a>1/3$ there exists $(F_1,F_2)\in[1/2,1]\times[1/2,1]$ s.t. $F_\mathrm{BDS}<(F_1+F_2)/2$.
        \item For $a>1/2$ there does not exist $(F_1,F_2)\in[1/2,1]\times[1/2,1]$ s.t. $F_\mathrm{BDS}\geq (F_1+F_2)/2$.
    \end{enumerate}
\end{proposition}
The above proposition characterizes the performance for the recurrence EPP with varying input states by providing the transition points of $a$ where the guaranteed improvement of output fidelity w.r.t. the fidelity of the AIS starts to break down (at $a=1/3$, with the Werner states as an example which corresponds to this point) and where any improvement completely vanishes (at $a=1/2$). Moreover, by noticing that the parameter $a$ represents the relative strength of the Pauli $Z$ error that the Bell pair experiences, this result again demonstrates the effect of error detection capability of the recurrence purification protocol by showing that when the Pauli $Z$ error in the input increases, the obtainable improvement from purification diminishes. 

In addition, when $a>0$ there is no guaranteed fidelity improvement w.r.t. the higher fidelity of the two input states. It is sufficient to find $(F_1,F_2)$ s.t. the successful output fidelity is lower than the higher input fidelity for each $a>0$. A simple example is to consider when one input state is noiseless, e.g. $F_1=1,F_2=1/2$, then the difference between the successful output fidelity and the higher input fidelity is $\frac{1}{1+a} - 1$, which is obviously negative for $a<0$.

\section{Optimal Time}\label{sec:opt_time}
Now we embed all the processes, including entanglement generation, EPP, and memory decoherence, into a unified timeline, as illustrated in Fig.~\ref{fig:opt_t_scenario}. We assume that the first Bell pair is generated at time 0, and the second Bell pair is obtained at time $t_1\geq 0$ due to the probabilistic nature of entanglement generation. The first Bell pair undergoes decoherence starting at time 0. We further assume that only \textit{one} Bell pair will be utilized later at $t_2\geq t_1$, and the EPP can be performed at time $t\in[t_1,t_2]$ in the time interval. Both generated Bell pairs will undergo decoherence between time $t_1$ and $t$, and the purified output state will keep decohering from $t$ to $t_2$. 

\begin{figure}[t]
    \centering
    \includegraphics[width=\columnwidth]{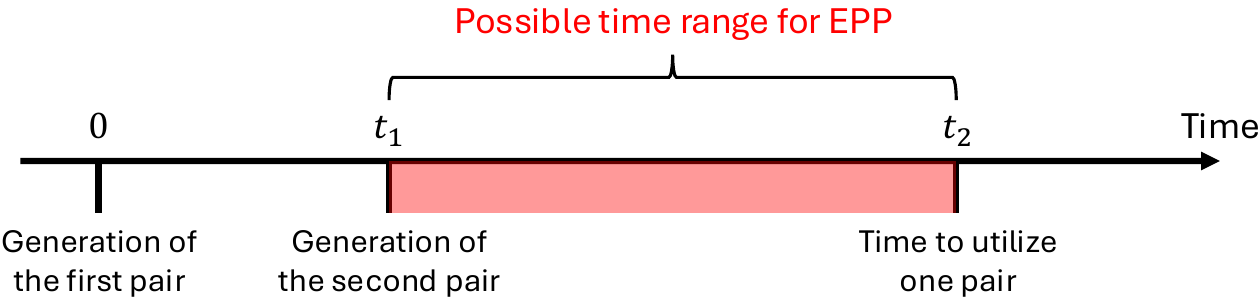}
    \caption{Visualization of EPP time optimization.}
    \label{fig:opt_t_scenario}
\end{figure}

Before presenting the technical results, we define the key system parameters. The \textit{raw fidelity} $F_0$ of entangled states immediately after successful generation is assumed to be a constant determined by details of the underlying physical system. Furthermore, we assume standardized quantum memories so that we have \textit{identical single qubit error channels}. The results in this section are robust against reasonable parameter fluctuation, i.e. as long as the raw fidelities do not deviate too much among different successful generation attempts and the memory decoherence rates are not too different from memory to memory, the optimal time results still hold, as demonstrated in Appendix~\ref{sec:robust_opt_time}.

\subsection{Rank-2 BDS under a single Pauli error}\label{sec:opt_time_rank2}
We first consider that the initial noisy Bell pair is a rank-2 BDS, and quantum memories are also subject to the same Pauli channel corresponding to the single error component in the rank-2 BDS which can be detected by the EPP. Without loss of generality, we consider the $X$ error, so that the raw Bell pair is $\rho=F_0\Phi^+ + (1-F_0)\Psi^+$ which we call a ``bit-flipped Bell state'', and memories are subject to identical bit-flip channels with rate $\kappa$, so that the fidelity evolves over time as $F_\mathrm{bit-flip}(F_0, t) = F_0e^{-2\kappa t} + \frac{1-e^{-2\kappa t}}{2}$. 

Our first figure of merit is the fidelity of the Bell pair which is utilized at $t_2$. Based on analytical BDS dynamics derived in Appendix~\ref{sec:bds_dynamics}, we obtain the following result.
\begin{proposition}\label{thm:rank2_fid_opt_time}
    Consider the quantum network scenario with bit-flipped Bell state under memory bit-flip channel. To obtain the highest fidelity Bell state to utilize conditioned on successful purification, the recurrence EPP should always be performed at the latest possible time.
\end{proposition}
Since all other figures of merit considered in this work monotonically increase when fidelity increases, the above result implies that, under the same considered error models, to obtain a Bell pair with the highest concurrence, (logarithmic) negativity, or distillable entanglement, the optimal time for the EPP is also the lastest possible time.

Proposition~\ref{thm:rank2_fid_opt_time} implies that for all bit-flipped Bell state pairs (both with fidelity above 1/2) and for all durations of memory decoherence, if the pairs are first purified and then undergo decoherence (FPTD), the final fidelity will be \textit{always lower than} in the case if both pairs first undergo decoherence and are purified at the end (FDTP). We can show that the fidelity difference between FDTP and FPTD ($F_\mathrm{FDTP}-F_\mathrm{FPTD}$) with decoherence duration $t$ is $$\frac{(2F_1-1)(2F_2-1)(F_1+F_2-1)\sinh(2\kappa t)}{(2F_1F_2-F_1-F_2+1)\left[e^{4\kappa t} + (2F_1-1)(2F_2-1)\right]},$$ which is greater than zero for all $t\geq 0$ and $F_1,F_2\geq 1/2$. These results demonstrate the effect of error detection capability of the recurrence EPP. 

Additionally, according to the monotonicity of entanglement measures with BDS fidelity, the above result can be generalized to figures of merit beyond fidelity. When we change the figure of merit to concurrence, (logarithmic) negativity, or distillable entanglement at $t_2$ conditioned on successful purification, the optimal time is also the latest possible time. Note that we ignore the two-way classical communication time in recurrence purification, and it will be omitted in the following as well to facilitate analytical studies. The effect of the classical communication delay is left for future work.

In the above we have not considered the probabilistic success of entanglement purification which affects the effective rate of entanglement distribution, because in our scenario if the purification fails there is no Bell pair to utilize at $t_2$. Hence, in order to incorporate the effect of purification success probability, we consider a \textit{normalized} figure of merit $\tilde{f}$ which is defined as the product of the success probability of purification performed at $t$ and the original figure of merit at $t_2$ conditioned on successful purification, i.e. $\tilde{f}=p_\mathrm{succ}(t)f(t_2)$. The definition of the normalized quantity aims at incorporating the trade-off between the output state quality and the success probability of getting the output state, similar to the spirit of figures of merit that include the network link active probability defined in Ref.~\cite{khatri2021towards,khatri2021policies,khatri2022design}.

We start with examining normalized fidelity and concurrence (negativity) which is linear in fidelity, and discover that in contrast to Prop.~\ref{thm:rank2_fid_opt_time} the optimal time is the earliest possible one.
\begin{proposition}\label{thm:rank2_normfid_opt_time}
    Consider the quantum network scenario with bit-flipped Bell state under memory bit-flip channel. To obtain the Bell state with the highest normalized fidelity, concurrence or negativity at $t_2$, the recurrence EPP should always be performed at the earliest possible time.
\end{proposition}

We also examine normalized distillable entanglement and logarithmic negativity which are non-linear in fidelity, and observe that unlike for normalized fidelity and concurrence (negativity), the optimal time is the earliest possible one for normalized logarithmic negativity, while for normalized distillable entanglement it is the latest possible time. Such a stark difference is due to the opposite convexity (concavity) of the two specific figures of merit as functions of fidelity.
\begin{proposition}\label{thm:norm_DELN_rank2_opt_time}
    Consider the quantum network scenario with bit-flipped Bell state under memory bit-flip channel. To obtain Bell state with the highest normalized distillable entanglement at $t_2$, the recurrence EPP should always be performed at the latest possible time. To obtain Bell state with the highest normalized logarithmic negativity at $t_2$, the recurrence EPP should always be performed at the earliest possible time.
\end{proposition}

\subsection{Werner state under depolarizing channel}\label{sec:opt_time_werner}
In this section we consider that the initial Bell pair is in the Werner form and both quantum memories undergo an identical depolarizing channel with rate $\kappa$, so that the fidelity dynamics is $F_\mathrm{Werner}(F_0,t) = F_0 e^{-2\kappa t} + \frac{1-e^{-2\kappa t}}{4}$. Initially we study the optimal time to obtain the highest fidelity at $t=t_2$ conditioned on successful purification. Unlike for the single-Pauli-error case above, the optimal time is parameter dependent (especially on $t_1$), and we obtain its explicit expression.
\begin{proposition}\label{thm:werner_fid_opt_time}
    Consider the quantum network scenario with Werner state under memory depolarizing channel. To obtain the highest fidelity Bell state at $t_2$, the recurrence EPP should always be performed at the earliest possible time if $t_1\geq t_1^*$, and should be performed at at $t^*$ if $t_1<t_1^*$, where
    \begin{align}
        & t_1^* = \frac{1}{2\kappa}\ln\left[\frac{3(4F_0-1)}{20 - 4F_0 - 16F_0^2}\right],\\
        & t^* = \frac{1}{2\kappa}\ln\left[\frac{(4F_0-1)^2e^{2\kappa t_1} + 3(4F_0-1)(1 + e^{2\kappa t_1})}{18}\right].
    \end{align}
\end{proposition}
Again, according to the argument that the considered entanglement measures monotonically increase with fidelity, the above optimal time also applies to concurrence, (logarithmic) negativity, or distillable entanglement.

Similar to the previous study of the rank-2 BDS, we take the success probability into account and evaluate the normalized fidelity and concurrence (negativity) at $t_2$, and see that the optimal time is also the earliest possible one. 
\begin{proposition}\label{thm:werner_normfid_opt_time}
    Consider the quantum network scenario with the Werner state under memory depolarizing channel. To obtain the Bell state with the highest normalized fidelity, concurrence or negativity at $t_2$, the recurrence EPP should always be performed at the earliest possible time. 
\end{proposition}

For normalized quantities we skip distillable entanglement as the exact expression of Werner state's distillable entanglement is not yet known, and we focus on logarithmic negativity only. It turns out that the dynamics of the normalized logarithmic negativity with the EPP time $t$ is not always monotonic, and the closed form expression of the optimal time is difficult to obtain due to the transcendental nature of the involved functions. Therefore, we visualize the optimal time instead in Fig.~\ref{fig:opt_puri_werner_normEN}. We see that if our objective is to obtain the highest normalized entanglement negativity, then for most practical $t_1$ and $t_2$ pairs we should purify immediately. However, if the time to utilize the Bell pair $t_2$ is late, we may want to perform the EPP after waiting for some time after the generation of the second entangled pair.
\begin{figure}[t]
    \centering
    \includegraphics[width=0.9\linewidth]{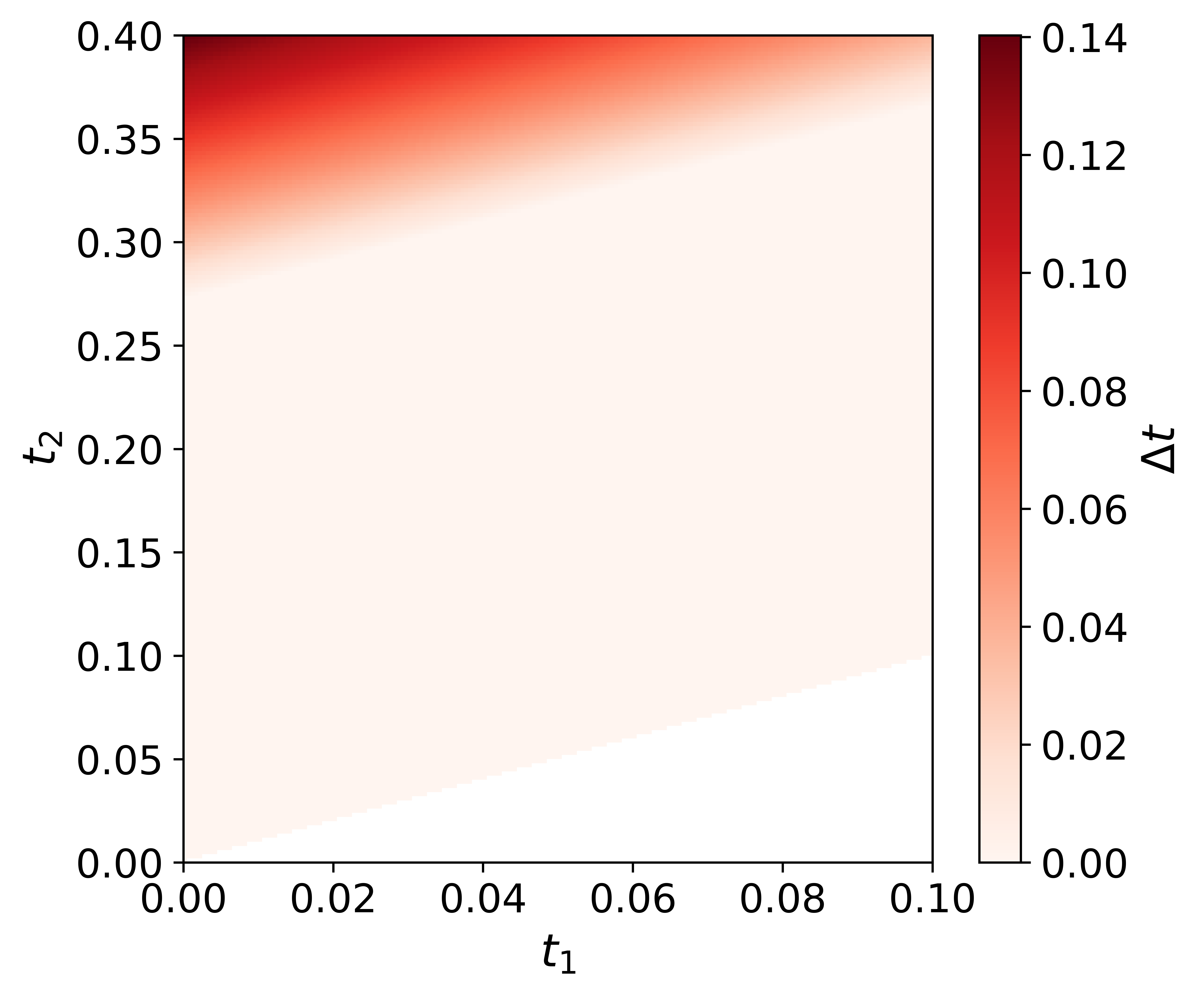}
    \caption{Optimal times for normalized entanglement negativity $\tilde{E}_N(t)$ under depolarizing channel with Werner initial states. Color represents the difference between the optimized EPP time $t$ and the start time $t_1$, i.e. $\Delta t=t-t_1$.}
    \label{fig:opt_puri_werner_normEN}
\end{figure}

\subsection{Optimal time transition}
\begin{figure*}[t]
    \centering
    \includegraphics[width=\linewidth]{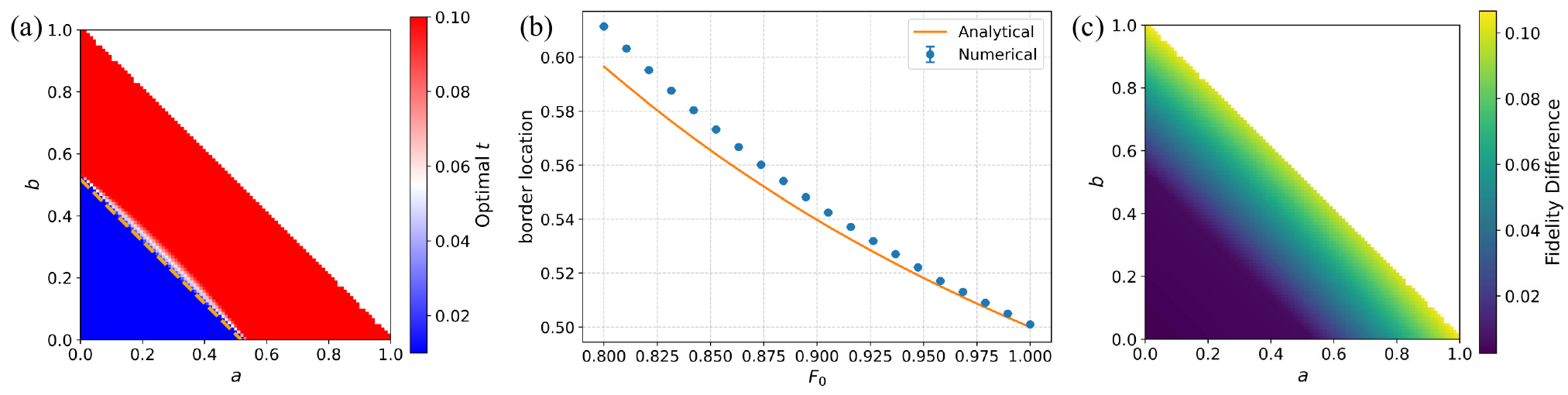}
    \caption{Numerical results for general continuous-time Pauli channels. (a) Optimal time for fidelity at $t_2$ conditioned on success under different Pauli memory decoherence channels. The optimal time $t$ is mapped to colors: blue represents early optimal time (close to $t_1$) and red represents late optimal time (close to $t_2$). The orange dashed line corresponds to the analytical estimation of transition border location Eqn.~\ref{eqn:transit_border}. We consider $t_1 = 0.01\gamma^{-1}$ and $t_2 = 0.1\gamma^{-1}$ in this case. (b) Comparison between the analytical approximation of optimal time transition border location and the numerically obtained transition border location for different raw fidelity $F_0$ of the initial Werner state. (c) Difference between the fidelity of the output state at $t_2$ after a successful purification at the optimal time and the fidelity of the later generated pair at $t_2$ assuming no EPP. Parameters $t_1,t_2,F_0$ are the same as for (a). Green means that performing the EPP leads to higher fidelity at $t_2$.}
    \label{fig:opt_time_transit}
\end{figure*}

We have studied two specific memory decoherence models. Here we generalize the discussion to an arbitrary continuous-time Pauli channel which describes memory decoherence. We present results for Werner initial states, while the phenomenon is general for different BDS as the initial states. 

To demonstrate how the optimal time for the EPP may depend on memory decoherence error models, we fix the following parameters, $t_1=0.01\gamma^{-1}$, $t_2=0.1\gamma^{-1}$, and the raw fidelity $F_0=0.95$, where $\gamma$ is the memory decoherence rate~\footnote{Here we use $\gamma$ to denote decoherence rate of a continuous-time Pauli channel with an arbitrary error pattern, different from the previous $\kappa$ for bit-flip and depolarizing channels, because of subtle differences in parametrization. See Appendix~\ref{sec:bds_dynamics} for the discussion on the equivalence of different ways of parametrization for bit-flip and depolarizing channels.}. According to recent experimental studies on various physical platforms including NV centers~\cite{hermans2022qubit}, neutral atoms~\cite{singh2023mid}, and rare-earth ions~\cite{ranvcic2018coherence}, which are all considered suitable for distributed quantum information processing, the reported qubit (memory) coherence times have reached the order of hundreds of milliseconds. Thus if we assume the coherence time of quantum memories for future quantum network implementations to be on the order of 1 second, $t_1=0.01\gamma^{-1}$ means that the second pair is generated roughly 10 milliseconds after the generation of the first pair. This time range can support approximately tens of consecutive heralded entanglement generation trials over a distance of 20km between repeater nodes, thus justifying the choice of the timing parameters.

Assuming both quantum memories undergo an identical Pauli channel with rate $\gamma$, we numerically calculate the optimal time to maximize the fidelity at $t_2$ conditioned on successful EPP, for all possible Pauli channels with different $X$ and $Y$ error rates, i.e. $\gamma_x=a\gamma$ and $\gamma_y=b\gamma$, s.t. $a,b\geq 0$ and $a+b\leq 1$. Note that the $Z$ error rate can be determined as $\gamma_z = (1-a-b)\gamma$ according to normalization. For numerical calculations, we always treat time in the unit of $\gamma^{-1}$ so that the quantities stay dimensionless. This treatment also make our observations independent of the absolute decoherence rate.

The optimal times for different $(a,b)$ pairs are visualized in Fig.~\ref{fig:opt_time_transit}(a), where only the lower triangle is a valid parameter space because $a+b\leq 1$. A remarkable phenomenon is observed. When moving from the upper right to the lower left of the $(a,b)$ plane, the optimal time experiences a drastic transition from the latest possible time to the earliest possible time, where the border of transition is demonstrated by the thin white stripe along the anti-diagonal direction. We call this phenomenon the \textit{optimal time transition}. 

Notice that the amount of the $X$ error from memory decoherence decreases while the $Z$ error increases when we move from the upper right to the lower left. Therefore, we could qualitatively argue that when the amount of the $X$ error is large the optimal time for fidelity at the end conditioned on success tends to be the latest possible time. In fact, this phenomenon is compatible with Prop.~\ref{thm:rank2_fid_opt_time} where we have proved that when the noise in the input states and memories is just the Pauli $X$, the optimal time for fidelity conditioned on success is always the latest possible time, independent of $t_1,t_2,F_0$. The existence of the optimal time transition re-emphasizes the nature of the error detecting capability of the recurrence EPP, while also revealing the recurrence EPP's sensitivity to the memory decoherence model.

After observing the clear border of transition, we would like to know where the border is located. Intuitively, the optimum being either the latest possible time or the earliest possible time suggests the fidelity at $t_2$ conditioned on success be a monotonically increasing or decreasing function with the EPP time $t$. Thus, the existence of the transition implies that for certain Pauli channels the leading order term that is linear in $t$ vanishes in $F_\mathrm{succ}(t_2)$. Following this direction, we are able to find a $t_1,t_2$-independent approximate expression of the location of the transition border $B$, which only depends on the raw fidelity $F_0$
\begin{equation}\label{eqn:transit_border}
    B\approx \frac{8F_0^2 -4F_0 + 5}{20F_0^2 - 4F_0 + 2}.
\end{equation}
The detailed derivation is provided in Appendix~\ref{sec:approx_border}. We can examine the accuracy of this approximation by comparing the analytical result and the numerically determined transition border location for different $F_0$, where the numerical results are averaged for different choices of $t_1\in[0,0.04\gamma^{-1}]$ and $t_2\in[0.06\gamma^{-1},0.1\gamma^{-1}]$. The comparison is demonstrated in Fig.~\ref{fig:opt_time_transit}(b), where the error bars associated with the points are too small to be visible. We see a good agreement for a wide range of the raw fidelity, while for higher $F_0$ the agreement is better, which validates the approximate analytical formula.

Finally, we consider another strategy which is to simply discard the earlier generated pair, and then at $t_2$ the state to utilize is just the later generated pair that has undergone a memory decoherence for $t_2-t_1$ time. We compare the fidelities at $t_2$ both with and without EPP performed at the optimal time, and the result is visualized in Fig.~\ref{fig:opt_time_transit}(c). We see that for the specific choice of parameters in this section, no matter what the Pauli channel is, performing EPP at the optimal time demonstrated in Fig.~\ref{fig:opt_time_transit}(a) will always provide higher fidelity at $t_2$, conditioned on successful EPP. This observation justifies the importance of determining the optimal time for different Pauli channels. Otherwise performing the EPP will potentially lead to worse result than simply discarding the older pair. 

\section{Conclusion and Discussion} \label{sec:discussion}
In summary, we examine the properties of EPPs in practical quantum network scenarios by focusing on two aspects: guaranteed improvement through successful purification in comparison to the non-identical input states, and optimal time for EPP when input states are generated successively under quantum memory decoherence. We analytically prove guaranteed improvement of practically important figures of merit for different forms of noisy input states, while the baseline varies with the error pattern of the input states. We also analytically prove parameter-independent optimal time for EPP, which, however, depends on the model of quantum memory decoherence. We further numerically demonstrate the phenomenon of optimal time transition when the quantum memory decoherence model changes, and analytically derive the approximate location for the transition border.

Our results in multiple ways deepen the understanding of EPP properties as well as offer insights in the optimization of EPP usage in realistic quantum networks. We would like to first emphasize that the guaranteed improvement explicitly demonstrates the universality property of the EPP for different classes of input states~\cite{zang2024no}, and the parameter independent optimal time for different memory decoherence models illustrates the impact of EPP universality on the desirable strategy of using EPP in quantum network scenarios. The optimal time of EPP can also potentially guide the optimization of EPP policy for realistic quantum networks~\cite{khatri2021towards,khatri2021policies,khatri2022design} with clear analytical intuition. Moreover, the optimal time transition phenomenon strongly suggests the importance of accurate characterization of noise models~\cite{martinis2015qubit,harper2020efficient,eisert2020quantum,kliesch2021theory} in quantum devices, especially quantum memories in our quantum network setup, as otherwise the optimization of protocol operation policies could lead to completely different solutions. Additionally, our results contribute to the recent analytical studies of quantum network scenarios~\cite{davies2024tools,davies2024entanglement,goodenough2024noise,zang2024analytical,inesta2025entanglement}, and demonstrate a new instance of provably optimal protocol/strategy for quantum networking systems~\cite{khatri2021policies,inesta2023optimal}. In the future, more detailed studies can be done by going beyond Pauli channels, removing the assumption of either $\mathrm{SU}(2)$ or Pauli twirling, and including errors in gate and measurement of the EPP, in more complex quantum network scenarios.

\begin{acknowledgments}
We acknowledge support from the NSF Quantum Leap Challenge Institute for Hybrid Quantum Architectures and Networks (NSF Award 2016136). T.Z. would like to acknowledge the support from the Marshall and Arlene Bennett Family Research Program. X.-A.C. and E.C. are supported by the U.S. Department of Energy Office of Science National Quantum Information Science Research Centers.
\end{acknowledgments}

\bibliography{references}

\onecolumngrid
\appendix

\section{Bell diagonal state dynamics under continuous-time separable Pauli channel}\label{sec:bds_dynamics}
In this appendix we derive analytical expression for how BDS evolves under typical single-qubit error channels on both qubits, which are used in the study of the optimal time. The results may be also applicable to other studies, for instance quantum networking simulation and performance evaluation~\cite{zang2023entanglement,chung2025cross,liu2025hardware}.

\subsection{Continuous-time general Pauli channel}
Pauli channels are defined as $\mathcal{E}_P(\rho) = \sum_{i=0}^3\alpha_iP_i\rho P_i$, where $P_0$ is the identity operator and $P_i$ are Pauli operators for $i=1,2,3$. We first derive continuous time Pauli channel from the channel for a single time step $\Delta t$
\begin{equation}
    \mathcal{E}_{\Delta t}(\rho) = (1-p)\rho + p(aX\rho X + bY\rho Y + cZ\rho Z),
    \label{eqn:pauli_def_discrete}
\end{equation}
where $p$ is the probability of having any Pauli error $\rho$ during time step $\Delta t$, and $a+b+c=1$ correspond to the relative strength of the three Pauli errors. A single application of this time step channel to qubit state $\rho$ can be represented by a $4\times 4$ matrix acting on a 4-dimensional vector which is the linearization of a $2\times 2$ density matrix $|\rho\rrangle=(\rho_{00},\rho_{01},\rho_{10},\rho_{11})^\mathrm{T}$. Explicitly we have
\begin{equation}
    \tilde{\mathcal{E}}_{\Delta t}|\rho\rrangle = 
    \begin{pmatrix}
    1-(a+b)p & 0 & 0 & (a+b)p\\
    0 & 1-(1+c)p & (a-b)p & 0\\
    0 & (a-b)p & 1-(1+c)p & 0\\
    (a+b)p & 0 & 0 & 1-(a+b)p\\
    \end{pmatrix}
    \begin{pmatrix}
    \rho_{00}\\
    \rho_{01}\\
    \rho_{10}\\
    \rho_{11}
    \end{pmatrix}.
\end{equation}
Then $n$ consecutive applications of the time step channel result in
\begin{align}
    \left(\tilde{\mathcal{E}}_{\Delta t}\right)^n|\rho\rrangle =& 
    \begin{pmatrix}
    \frac{1 + (1-2ap-2bp)^n}{2} & 0 & 0 & \frac{1 - (1-2ap-2bp)^n}{2}\\
    0 & \frac{(1-2bp-2cp)^n + (1-2ap-2cp)^n}{2} & \frac{(1-2bp-2cp)^n - (1-2ap-2cp)^n}{2} & 0\\
    0 & \frac{(1-2bp-2cp)^n - (1-2ap-2cp)^n}{2} & \frac{(1-2bp-2cp)^n + (1-2ap-2cp)^n}{2} & 0\\
    \frac{1 - (1-2ap-2bp)^n}{2} & 0 & 0 & \frac{1 + (1-2ap-2bp)^n}{2}\\
    \end{pmatrix}
    \begin{pmatrix}
    \rho_{00}\\
    \rho_{01}\\
    \rho_{10}\\
    \rho_{11}
    \end{pmatrix}\nonumber\\
    =&
    \begin{pmatrix}
    \frac{1 + e^{-2(\gamma_x+\gamma_y)t}}{2} & 0 & 0 & \frac{1 - e^{-2(\gamma_x+\gamma_y)t}}{2}\\
    0 & e^{-2\gamma_zt}\frac{e^{-2\gamma_yt} + e^{-2\gamma_xt}}{2} & e^{-2\gamma_zt}\frac{e^{-2\gamma_yt} - e^{-2\gamma_xt}}{2} & 0\\
    0 & e^{-2\gamma_zt}\frac{e^{-2\gamma_yt} - e^{-2\gamma_xt}}{2} & e^{-2\gamma_zt}\frac{e^{-2\gamma_yt} + e^{-2\gamma_xt}}{2} & 0\\
    \frac{1 - e^{-2(\gamma_x+\gamma_y)t}}{2} & 0 & 0 & \frac{1 + e^{-2(\gamma_x+\gamma_y)t}}{2}\\
    \end{pmatrix}
    \begin{pmatrix}
    \rho_{00}\\
    \rho_{01}\\
    \rho_{10}\\
    \rho_{11}
    \end{pmatrix}\nonumber\\
    =& 
    \begin{pmatrix}
    \frac{1 + e^{-2(\gamma_x+\gamma_y)t}}{2}\rho_{00} + \frac{1 - e^{-2(\gamma_x+\gamma_y)t}}{2}\rho_{11}\\
    e^{-2\gamma_zt}\frac{e^{-2\gamma_yt} + e^{-2\gamma_xt}}{2}\rho_{01} + e^{-2\gamma_zt}\frac{e^{-2\gamma_yt} - e^{-2\gamma_xt}}{2}\rho_{10}\\
    e^{-2\gamma_zt}\frac{e^{-2\gamma_yt} - e^{-2\gamma_xt}}{2}\rho_{01} + e^{-2\gamma_zt}\frac{e^{-2\gamma_yt} + e^{-2\gamma_xt}}{2}\rho_{10}\\
    \frac{1 - e^{-2(\gamma_x+\gamma_y)t}}{2}\rho_{00} + \frac{1 + e^{-2(\gamma_x+\gamma_y)t}}{2}\rho_{11}
    \end{pmatrix}
    = \tilde{\mathcal{E}}(n\Delta t)|\rho\rrangle = \tilde{\mathcal{E}}(t)|\rho\rrangle,
\end{align}
where for the second equality we have used the limit $\lim_{x\rightarrow\infty}(1+1/x)^x=e$, defined error rate $\gamma\Delta t=p$ and $\gamma_x=a\gamma,\gamma_y=b\gamma,\gamma_z=c\gamma$. We can express the resulting qubit state explicitly in a matrix form
\begin{equation}
    \rho(t) =
    \begin{pmatrix}
    \frac{1 + e^{-2(\gamma_x+\gamma_y)t}}{2}\rho_{00} + \frac{1 - e^{-2(\gamma_x+\gamma_y)t}}{2}\rho_{11} & e^{-2\gamma_zt}\frac{e^{-2\gamma_yt} + e^{-2\gamma_xt}}{2}\rho_{01} + e^{-2\gamma_zt}\frac{e^{-2\gamma_yt} - e^{-2\gamma_xt}}{2}\rho_{10}\\
    e^{-2\gamma_zt}\frac{e^{-2\gamma_yt} - e^{-2\gamma_xt}}{2}\rho_{01} + e^{-2\gamma_zt}\frac{e^{-2\gamma_yt} + e^{-2\gamma_xt}}{2}\rho_{10} &
    \frac{1 - e^{-2(\gamma_x+\gamma_y)t}}{2}\rho_{00} + \frac{1 + e^{-2(\gamma_x+\gamma_y)t}}{2}\rho_{11}
    \end{pmatrix}.
    \label{eqn:general_pauli_result}
\end{equation}
From this we can see the contributions from each single Pauli channel
\begin{equation}
\begin{aligned}
    &\rho(t)\\
    =& \frac{1}{4}\left(1 - e^{-2(\gamma_x+\gamma_y)t} + e^{-2(\gamma_y+\gamma_z)t} - e^{-2(\gamma_x+\gamma_z)t}\right)X\rho(0)X + \frac{1}{4}\left(1 - e^{-2(\gamma_x+\gamma_y)t} - e^{-2(\gamma_y+\gamma_z)t} + e^{-2(\gamma_x+\gamma_z)t}\right)Y\rho(0)Y\\
    &+ \frac{1}{4}\left(1 + e^{-2(\gamma_x+\gamma_y)t} - e^{-2(\gamma_y+\gamma_z)t} - e^{-2(\gamma_x+\gamma_z)t}\right)Z\rho(0)Z + \frac{1}{4}\left(1 + e^{-2(\gamma_x+\gamma_y)t} + e^{-2(\gamma_y+\gamma_z)t} + e^{-2(\gamma_x+\gamma_z)t}\right)\rho(0).
\end{aligned}
\end{equation}
That is we have the Pauli channel for duration $t$
\begin{equation}
\begin{aligned}
    \tilde{\mathcal{E}}_P(t) =& \frac{1 + e^{-2(\gamma_x+\gamma_y)t} + e^{-2(\gamma_y+\gamma_z)t} + e^{-2(\gamma_x+\gamma_z)t}}{4}\mathrm{Id} + \frac{1 - e^{-2(\gamma_x+\gamma_y)t} + e^{-2(\gamma_y+\gamma_z)t} - e^{-2(\gamma_x+\gamma_z)t}}{4}\mathrm{X}\\
    &+ \frac{1 - e^{-2(\gamma_x+\gamma_y)t} - e^{-2(\gamma_y+\gamma_z)t}+ e^{-2(\gamma_x+\gamma_z)t}}{4}\mathrm{Y} + \frac{1 + e^{-2(\gamma_x+\gamma_y)t} - e^{-2(\gamma_y+\gamma_z)t} - e^{-2(\gamma_x+\gamma_z)t}}{4}\mathrm{Z}\\
    =& p_I(t)\mathrm{Id} + p_X(t)\mathrm{X} + p_Y(t)\mathrm{Y} + p_Z(t)\mathrm{Z},
    \label{eqn:p_i(t)_def}
\end{aligned}
\end{equation}
where we define $p_i(t),\ i=I,X,Y,Z$, and we use the upright font to represent the channel of applying one Pauli operator to a density matrix.

\subsection{Bell diagonal states (BDS)}
Let us consider general BDS $\rho_\mathrm{BDS} = \lambda_1\Psi^+ + \lambda_2\Psi^- + \lambda_3\Phi^+ + \lambda_4\Phi^-$ with the normalization condition $\sum_{i=1}^3\lambda_i=1$, and $\lambda_1=F$ is the fidelity of $\rho$ w.r.t. $|\psi^+\rangle$. Here we demonstrate how the general BDS will vary over time under general Pauli channels. The channel applied to the bipartite state is a tensor product of two independent Pauli channels on both memories
\begin{equation}
\begin{aligned}
    &\tilde{\mathcal{E}}_{P,AB}(t) = \tilde{\mathcal{E}}_{P,A}(t)\otimes\tilde{\mathcal{E}}_{P,B}(t)\\
    =& \left(p_I^{(A)}(t)\mathrm{Id}_A + p_X^{(A)}(t)\mathrm{X}_A + p_Y^{(A)}(t)\mathrm{Y}_A + p_Z^{(A)}(t)\mathrm{Z}_A\right)\otimes\left(p_I^{(B)}(t)\mathrm{Id}_B + p_X^{(B)}(t)\mathrm{X}_B + p_Y^{(B)}(t)\mathrm{Y}_B + p_Z^{(B)}(t)\mathrm{Z}_B\right)\\
    =& p_I^{(A)}(t)p_I^{(B)}(t)\mathrm{Id}_A\otimes\mathrm{Id}_B + p_X^{(A)}(t)p_X^{(B)}(t)\mathrm{X}_A\otimes\mathrm{X}_B + p_Y^{(A)}(t)p_Y^{(B)}(t)\mathrm{Y}_A\otimes\mathrm{Y}_B + p_Z^{(A)}(t)p_Z^{(B)}(t)\mathrm{Z}_A\otimes\mathrm{Z}_B\\
    &+ p_I^{(A)}(t)p_X^{(B)}(t)\mathrm{Id}_A\otimes\mathrm{X}_B + p_I^{(A)}(t)p_Y^{(B)}(t)\mathrm{Id}_A\otimes\mathrm{Y}_B + p_I^{(A)}(t)p_Z^{(B)}(t)\mathrm{Id}_A\otimes\mathrm{Z}_B + p_X^{(A)}(t)p_I^{(B)}(t)\mathrm{X}_A\otimes\mathrm{Id}_B\\
    &+ p_X^{(A)}(t)p_Y^{(B)}(t)\mathrm{X}_A\otimes\mathrm{Y}_B + p_X^{(A)}(t)p_Z^{(B)}(t)\mathrm{X}_A\otimes\mathrm{Z}_B + p_Y^{(A)}(t)p_I^{(B)}(t)\mathrm{Y}_A\otimes\mathrm{Id}_B + p_Y^{(A)}(t)p_X^{(B)}(t)\mathrm{Y}_A\otimes\mathrm{X}_B\\
    &+ p_Y^{(A)}(t)p_Z^{(B)}(t)\mathrm{Y}_A\otimes\mathrm{Z}_B + p_Z^{(A)}(t)p_I^{(B)}(t)\mathrm{Z}_A\otimes\mathrm{Id}_B + p_Z^{(A)}(t)p_X^{(B)}(t)\mathrm{Z}_A\otimes\mathrm{X}_B + p_Z^{(A)}(t)p_Y^{(B)}(t)\mathrm{Z}_A\otimes\mathrm{Y}_B,
    \label{eqn:Pauli_channel_eff}
\end{aligned}
\end{equation}
with $4\times 4=16$ terms in total. Note that single bilateral Pauli channels ($\mathrm{P}_i\otimes \mathrm{P}_i$) will not affect Bell states, and $\mathrm{P}_i = \mathrm{P}_j\circ\mathrm{P}_k = \mathrm{P}_k\circ\mathrm{P}_j$ for $i\neq j,j\neq k,i\neq k$ where $i,j,k\in\{x,y,z\}$. We then have
\begin{equation}
\begin{aligned}
    \left[\tilde{\mathcal{E}}_{P,AB}(t)\right](\Psi^+_{AB}) =& \left(p_I^{(A)}(t)p_I^{(B)}(t) + p_X^{(A)}(t)p_X^{(B)}(t)  + p_Y^{(A)}(t)p_Y^{(B)}(t) + p_Z^{(A)}(t)p_Z^{(B)}(t)\right)\Psi^+_{AB}\\
    &+ \left(p_I^{(A)}(t)p_X^{(B)}(t) + p_X^{(A)}(t)p_I^{(B)}(t) + p_Y^{(A)}(t)p_Z^{(B)}(t) + p_Z^{(A)}(t)p_Y^{(B)}(t)\right)\Phi^+_{AB}\\
    &+ \left(p_I^{(A)}(t)p_Y^{(B)}(t) + p_Y^{(A)}(t)p_I^{(B)}(t) + p_X^{(A)}(t)p_Z^{(B)}(t) + p_Z^{(A)}(t)p_X^{(B)}(t)\right)\Phi^-_{AB}\\
    &+ \left(p_I^{(A)}(t)p_Z^{(B)}(t) + p_Z^{(A)}(t)p_I^{(B)}(t)  + p_X^{(A)}(t)p_Y^{(B)}(t) + p_Y^{(A)}(t)p_X^{(B)}(t)\right)\Psi^-_{AB}\\
    =& C_I(t)\Psi^+_{AB} + C_X(t)\Phi^+_{AB} + C_Y(t)\Phi^-_{AB} + C_Z(t)\Psi^-_{AB}.
    \label{eqn:Psi+AB}
\end{aligned}
\end{equation}
According to Eqn.~\ref{eqn:Pauli_channel_eff}, we can evaluate the evolution of four Bell components separately similar to Eqn.~\ref{eqn:Psi+AB}. The effect of general Pauli channel on the other 3 Bell states is
\begin{equation}
\begin{aligned}
    \left[\tilde{\mathcal{E}}_{P,AB}(t)\right](\Psi^-_{AB}) =& C_I(t)\Psi^-_{AB} + C_X(t)\Phi^-_{AB} + C_Y(t)\Phi^+_{AB} + C_Z(t)\Psi^+_{AB},\\
    \left[\tilde{\mathcal{E}}_{P,AB}(t)\right](\Phi^+_{AB}) =& C_I(t)\Phi^+_{AB} + C_X(t)\Psi^+_{AB} + C_Y(t)\Psi^-_{AB} + C_Z(t)\Phi^-_{AB},\\
    \left[\tilde{\mathcal{E}}_{P,AB}(t)\right](\Phi^-_{AB}) =& C_I(t)\Phi^-_{AB} + C_X(t)\Psi^-_{AB} + C_Y(t)\Psi^+_{AB} + C_Z(t)\Phi^+_{AB}.\\
\end{aligned}
\end{equation}
Then we obtain the evolution of BDS under such a channel
\begin{equation}
\begin{aligned}
    &\left[\tilde{\mathcal{E}}_{P,AB}(t)\right](\rho_\mathrm{BDS})\\
    =& \lambda_1\left(C_I(t)\Psi^+_{AB} + C_X(t)\Phi^+_{AB} + C_Y(t)\Phi^-_{AB} + C_Z(t)\Psi^-_{AB}\right) + \lambda_2\left(C_I(t)\Psi^-_{AB} + C_X(t)\Phi^-_{AB} + C_Y(t)\Phi^+_{AB} + C_Z(t)\Psi^+_{AB}\right)\\
    &+ \lambda_3\left(C_I(t)\Phi^+_{AB} + C_X(t)\Psi^+_{AB} + C_Y(t)\Psi^-_{AB} + C_Z(t)\Phi^-_{AB}\right) + \lambda_4\left(C_I(t)\Phi^-_{AB} + C_X(t)\Psi^-_{AB} + C_Y(t)\Psi^+_{AB} + C_Z(t)\Phi^+_{AB}\right)\\
    =& \left(\lambda_1C_I(t) + \lambda_2C_Z(t) + \lambda_3C_X(t) + \lambda_4C_Y(t)\right)\Psi^+_{AB} + \left(\lambda_1C_Z(t) + \lambda_2C_I(t) + \lambda_3C_Y(t) + \lambda_4C_X(t)\right)\Psi^-_{AB}\\
    &+ \left(\lambda_1C_X(t) + \lambda_2C_Y(t) + \lambda_3C_I(t) + \lambda_4C_Z(t)\right)\Phi^+_{AB} + \left(\lambda_1C_Y(t) + \lambda_2C_X(t) + \lambda_3C_Z(t) + \lambda_4C_I(t)\right)\Phi^-_{AB}.
\end{aligned}
\end{equation}
The above transformation of four diagonal elements can be viewed from matrix multiplication perspective
\begin{equation}
\begin{aligned}
    \begin{pmatrix}
        \lambda_1'\\
        \lambda_2'\\
        \lambda_3'\\
        \lambda_4'\\
    \end{pmatrix} =& 
    \begin{pmatrix}
        p_I^{(A)}(t) & p_Z^{(A)}(t) & p_X^{(A)}(t) & p_Y^{(A)}(t) \\
        p_Z^{(A)}(t) & p_I^{(A)}(t) & p_Y^{(A)}(t) & p_X^{(A)}(t) \\
        p_X^{(A)}(t) & p_Y^{(A)}(t) & p_I^{(A)}(t) & p_Z^{(A)}(t) \\
        p_Y^{(A)}(t) & p_X^{(A)}(t) & p_Z^{(A)}(t) & p_I^{(A)}(t) \\        
    \end{pmatrix}
    \begin{pmatrix}
        p_I^{(B)}(t) & p_Z^{(B)}(t) & p_X^{(B)}(t) & p_Y^{(B)}(t) \\
        p_Z^{(B)}(t) & p_I^{(B)}(t) & p_Y^{(B)}(t) & p_X^{(B)}(t) \\
        p_X^{(B)}(t) & p_Y^{(B)}(t) & p_I^{(B)}(t) & p_Z^{(B)}(t) \\
        p_Y^{(B)}(t) & p_X^{(B)}(t) & p_Z^{(B)}(t) & p_I^{(B)}(t) \\        
    \end{pmatrix}
    \begin{pmatrix}
        \lambda_1\\
        \lambda_2\\
        \lambda_3\\
        \lambda_4\\
    \end{pmatrix}\\
    =& 
    \begin{pmatrix}
        C_I(t) & C_Z(t) & C_X(t) & C_Y(t) \\
        C_Z(t) & C_I(t) & C_Y(t) & C_X(t) \\
        C_X(t) & C_Y(t) & C_I(t) & C_Z(t) \\
        C_Y(t) & C_X(t) & C_Z(t) & C_I(t) \\        
    \end{pmatrix}
    \begin{pmatrix}
        \lambda_1\\
        \lambda_2\\
        \lambda_3\\
        \lambda_4\\
    \end{pmatrix},
\end{aligned}
\end{equation}
where the two $p(t)$ matrices in the first row commute, which demonstrates the independence of two memory decoherence processes. Note that the coefficients only correspond to errors that happened, and to convert to the exact order of the four Bell states in the main text we only need to apply a single qubit $X$ gate.

\subsection{Comment on parametrization of the bit-flip channel and the depolarizing channel}
The bit-flip channel can be canonically written as
\begin{align}
    \mathcal{E}_{\mathrm{bit-flip}}(\rho) = p\rho + (1-p)\frac{X\rho X + \rho}{2} = \frac{(1+p)}{2}\rho + \frac{(1+p)}{2}X\rho X,
\end{align}
where $(1-p)$ can be interpreted as the probability for the complete bit-flip channel to happen, and $(X\rho X + \rho)/2$ is the steady state of the bit-flip channel. Different from the previous microscopic treatment of general continous-time Pauli channels, we may also define bit-flip rate $\kappa$ with $p = \exp(-\kappa t)$. We use this parametrization in Sec.~\ref{sec:opt_time_rank2}. For the continuous-time Pauli channels characterized by rate $\gamma$, when $\gamma=\gamma_x$ the resulting channel should be equivalent to the bit-flip channel. According to Eqn.~\ref{eqn:p_i(t)_def} the resulting channel should be
\begin{align}
    \mathcal{E}_{\mathrm{bit-flip}}(\rho) = \frac{1+e^{-2\gamma t}}{2}\rho + \frac{1 - e^{-2\gamma t}}{2}X\rho X.
\end{align}
Comparing the definition of $\kappa$ and the above equation, we obtain the connection between the two ways of parametrization for the bit-flip channel: $\kappa_{\mathrm{bit-flip}} = 2\gamma$. This also applies to other 2 single-Pauli channels.

The depolarizing channel (for qubit) can be written as
\begin{align}
    \mathcal{E}_{\mathrm{dp}}(\rho) = p\rho + (1-p)\frac{I}{2} = p\rho + \frac{1-p}{4}(X\rho X + Y\rho Y + Z\rho Z + \rho),
\end{align}
where $(1-p)$ can be interpreted as the probability for the complete depolarizing channel to happen, and the maximally mixed state $I/2$ is the output of the complete depolarizing channel. Similar to the bit-flip channel, we can define the depolarizing rate $\kappa$ with $p = \exp(-\kappa t)$. We use this parametrization in Sec.~\ref{sec:opt_time_werner}. For the continuous-time Pauli channels characterized by rate $\gamma$, when $\gamma_x=\gamma_y=\gamma_z=\gamma/3$ the resulting channel should be equivalent to the depolarizing channel. According to Eqn.~\ref{eqn:p_i(t)_def} the resulting channel should be 
\begin{align}
    \mathcal{E}_{\mathrm{dp}}(\rho) = \frac{1 + 3e^{-4\gamma t/3}}{4}\rho + \frac{1 - e^{-4\gamma t/3}}{4}X\rho X + \frac{1 - e^{-4\gamma t/3}}{4}Y\rho Y + \frac{1 - e^{-4\gamma t/3}}{4}Z\rho Z.
\end{align}
Comparing the definition of $\kappa$ and the above equation, we obtain the connection between the two ways of parametrization for the depolarizing channel: $\kappa_{\mathrm{dp}} = 4\gamma/3$.

\section{Proofs for guaranteed improvement}\label{sec:proof_improv}
In this appendix, we provide detailed proofs of the propositions on guaranteed improvement from Sec~\ref{sec:improv}. 

\subsection{Proof of Proposition~\ref{thm:rank2bds_fidincr}}
\begin{proof}
    We construct the fidelity increase function
    \begin{equation}
        F_\mathrm{incr}^{(2)}(F_1,F_2) = \frac{F_1F_2}{F_1F_2 + (1-F_1)(1-F_2)} - \max\{F_1,F_2\}.
    \end{equation}
    Note the symmetry between $F_1,F_2$ in the defined function, without loss of generality, we can use $F_1$ to replace $\max\{F_1,F_2\}$ and arrive at a modified increase function
    \begin{equation}
        F_\mathrm{incr}^{(2)'}(F_1,F_2) = \frac{F_1F_2 - F_1^2F_2 - F_1(1-F_1)(1-F_2)}{F_1F_2 + (1-F_1)(1-F_2)}.
    \end{equation}
    Obviously, the denominator is positive on $(F_1,F_2)\in[1/2,1]\times[1/2,1]$, so to prove the statement we only need to focus on the numerator $N_\mathrm{incr}^{(2)'}=F_1F_2 - F_1^2F_2 - F_1(1-F_1)(1-F_2)$. For the numerator we have
    \begin{equation}
        N_\mathrm{incr}^{(2)'} = F_1(1-F_1)(2F_2-1) \geq 0,
    \end{equation}
    for $F_1,F_2\geq 1/2$. Hence, $F_\mathrm{incr}^{(2)}(F_1,F_2)\geq 0$ for $F_1,F_2\geq 1/2$.
\end{proof}

\subsection{Proof of Proposition~\ref{thm:werner_fidincr}}
\begin{proof}
    We first prove the statement for $\lambda=1/2$. We construct the fidelity increase function
    \begin{equation}
    \begin{aligned}
        F_\mathrm{incr}^{(W)}(F_1,F_2) =& \frac{F_1F_2 + \frac{(1-F_1)(1-F_2)}{9}}{F_1F_2 + \frac{F_1(1-F_2)+(1-F_1)F_2}{3} + \frac{5(1-F_1)(1-F_2)}{9}} - \frac{F_1 + F_2}{2}\\
        =& \frac{(2-8F_2)F_1^2 - (8F_2^2-24F_2+7)F_1 + 2F_2^2-7F_2+2}{2[5 - 2F_2 + F_1(8F_2-2))]}.
    \end{aligned}
    \end{equation}
    Obviously, the denominator is positive on $(F_1,F_2)\in[1/2,1]\times[1/2,1]$, so to prove the statement we only need to focus on the numerator $N_\mathrm{incr}^{(W)}(F_1,F_2) = (2-8F_2)F_1^2 + (-8F_2^2+24F_2-7)F_1 + (2F_2^2-7F_2+2)$. First observe that $2-8F_2<0$ for $F_2\in[1/2,1]$, so $N_\mathrm{incr}^{(W)}(F_1,F_2)$ will always be a parabola opening downward on the fidelity domain of our interest, when we view $F_1$ as the only argument and $F_2$ as a varying parameter from 1/2 to 1. We know that the minimum of a downward opening parabola on a closed domain will be obtained at the boundary of the domain, so we only need to check the value of $N_\mathrm{incr}^{(W)}(F_1,F_2)$ for $F_1=1/2$ and $F_1=1$. At the boundary we obtain the explicit expressions of the numerator for varying $F_2$
    \begin{equation}
        N_\mathrm{incr}^{(W)}(1/2,F_2) = -2F_2^2 + 3F_2 - 1,~N_\mathrm{incr}^{(W)}(1,F_2) = 3(-2F_2^2 + 3F_2 - 1).
    \end{equation}
    Thus by checking the above parabola of $F_2$ at $F_2=1/2,1$, $N_\mathrm{incr}^{(W)}(1/2,1/2)=N_\mathrm{incr}^{(W)}(1/2,1)=N_\mathrm{incr}^{(W)}(1,1/2)=N_\mathrm{incr}^{(W)}(1,1)=0$, we have
    \begin{equation}
        N_\mathrm{incr}^{(W)}(F_1,F_2)\geq \min\{N_\mathrm{incr}^{(W)}(1/2,1/2),N_\mathrm{incr}^{(W)}(1/2,1), N_\mathrm{incr}^{(W)}(1,1/2),N_\mathrm{incr}^{(W)}(1,1)\} = 0.   
    \end{equation}
    As $F_\mathrm{baseline}(\lambda,F_1,F_2)$ monotonically increases when $\lambda$ increases for all fixed $F_1,F_2$, the above result directly implies that $F^{(W)}(F_1,F_2)\geq F_\mathrm{baseline}(\lambda,F_1,F_2)$ for $\lambda\in[0,1/2]$.

    For $\lambda\in(1/2,1]$, we consider $F_1=1,F_2=1/2$ or $F_1=1/2,F_2=1$, and explicit calculation yields
    \begin{equation}
        F^{(W)}(1,1/2) - F_\mathrm{baseline}(\lambda,1,1/2) = \frac{3}{4}-\frac{3\lambda}{2}<0,
    \end{equation}
    which provides an instance of not guaranteed improvement, thus the proof is completed.
\end{proof}

\subsection{Fidelity threshold for higher baselines}
We compare successful output fidelity of the recurrence EPP given Werner state $F^{(W)}(F_1,F_2)$ and the convex combination of the higher input fidelity and the lower input fidelity $F_\mathrm{baseline}(\lambda,F_1,F_2)$, for $\lambda>1/2$.
\begin{proposition}
    For the successful output fidelity of the recurrence EPP given two input Werner states $\rho_\mathrm{W}(F_1)$ and $\rho_\mathrm{W}(F_2)$ with $F_1,F_2\geq 1/2$, we have $F^{(W)}(F_1,F_2)\geq F_\mathrm{baseline}(\lambda,F_1,F_2)$ for $F_1,F_2\in[\frac{\lambda}{2(1-\lambda)},1]$ and $\lambda\in[1/2,2/3]$.
\end{proposition}
\begin{proof}
    Given the symmetry between $F_1$ and $F_2$, we can focus on the case where $F_1\geq F_2$, and study the difference
    \begin{align}
        \left.F_{\mathrm{incr},\lambda}^{(W)}\right\vert_{F_1\geq F_2} =& F^{(W)}(F_1,F_2) - [\lambda F_1+(1-\lambda)F_2]\nonumber\\
        =& \frac{1 - (1+5r)F_1 + (5r-6)F_2 + 2rF_1^2 + 2(1-r)F_2^2 + 12F_1F_2 + 8(r-1)F_1F_2^2 - 8rF_1^2F_2}{5 - 2F_1 - 2F_2 + 8F_1F_2},
    \end{align}
    where the denominator is positive for $F_1,F_2\in[1/2,1]$, so we only need to examine the numerator $N_{\mathrm{incr},\lambda}^{(W)}$ for $1\geq F_1\geq F_2\geq 1/2$. First of all, we can show that $N_{\mathrm{incr},\lambda}^{(W)}$ is concave for fixed $F_2$, by taking the second order partial derivative with respect to $F_1$
    \begin{align}
        \frac{\partial^2}{\partial F_1^2}N_{\mathrm{incr},\lambda}^{(W)} = 4\lambda(1-4F_2) < 0.
    \end{align}
    Therefore, the minimal value of $N_{\mathrm{incr},\lambda}^{(W)}$ for fixed $F_2$ can only be obtained on the boundary, i.e. $F_1=1/2$ or $F_1=1$. Notice that at $F_1=1/2$ we have
    \begin{align}
        \left.N_{\mathrm{incr},\lambda}^{(W)}\right\vert_{F_1=1/2} = \frac{1}{2}(2F_2-1)[4\lambda-1-2(1-\lambda)F_2],
    \end{align}
    which is always non-negative for $F_2\in[1/2,1]$. Then we consider $F_1=1$
    \begin{align}
        \left.N_{\mathrm{incr},\lambda}^{(W)}\right\vert_{F_1=1} = 3(1-F_2)[2(1-\lambda)F_2-\lambda],
    \end{align}
    whose positivity is determined by the square bracket term alone. It is then easy to see that for $\lambda\in(2/3,1]$, there is no $F_2\in[1/2,1]$ s.t. $\left.N_{\mathrm{incr},\lambda}^{(W)}\right\vert_{F_1=1}\geq 0$. Meanwhile, as long as $\lambda\in[1/2,2/3]$, we have $\left.N_{\mathrm{incr},\lambda}^{(W)}\right\vert_{F_1=1}\geq 0$ for $F_2\in[\frac{\lambda}{2(1-\lambda)},1]$. Then according to the concavity of $N_{\mathrm{incr},\lambda}^{(W)}$ and the symmetry of $F_{\mathrm{incr},\lambda}^{(W)}$, the statement is proved.
\end{proof}

\subsection{Proof of Proposition~\ref{thm:werner_de_improv}}
\begin{proof}
We want to examine the following function
\begin{equation}
    E_{D}^\mathrm{low-up}(F_1,F_2) = I_{A\rangle B}^{W}(F'_W) - R(\frac{F_1+F_2}{2}).
\end{equation}
According to the symmetry between $F_1$ and $F_2$, we can transform the coordinates into $F_1=(F_1'-F_2')/\sqrt{2},F_2=(F_1'+F_2')/\sqrt{2}$. After this coordinate transformation, the input fidelity region becomes $F_1'\in[1/\sqrt{2},\sqrt{2}],F_2'\in[1/\sqrt{2}-F_1',F_1'-1/\sqrt{2}]$. We first examine the partial derivative w.r.t. $F_2'$ (anti-diagonal coordinate) of $E_{D}^\mathrm{low-up}(F_1',F_2')$
\begin{equation}\label{eqn:dist_ent_F2'deriv}
    \frac{\partial}{\partial F_2'}E_{D}^\mathrm{low-up} = \frac{6(2\sqrt{2}F_1'-7)\log\left[\frac{9(7-2\sqrt{2}F_1')}{F_2'^2 - F_1'(F_1'+\sqrt{2}) + 4}-15\right]}{[4F_2'^2 + 2F_1'(\sqrt{2}-2F_1')-5]^2}F_2'.
\end{equation}
Note that here the base of logarithm does not matter because it only contributes as an additional factor, as the constant terms in the definition of the coherent information and the Rains bound cancel each other. It is obvious that the denominator is non-negative, and the pre-factor of numerator $(2\sqrt{2}F_1'-7)$ is negative. Then we examine the positivity of the logarithmic term in the numerator. By evaluating the argument of the logarithmic term
\begin{align}
    &\frac{\partial}{\partial F_1'}\frac{9(7-2\sqrt{2}F_1')}{F_2'^2 - F_1'(F_1'+\sqrt{2}) + 4} = \frac{18F_1'(7-\sqrt{2}F_1') - 9\sqrt{2}(1+2F_2^2)}{[F_2'^2 - F_1'(\sqrt{2}+F_1') + 4]^2} > 0,\\
    &\frac{\partial}{\partial F_2'}\frac{9(7-2\sqrt{2}F_1')}{F_2'^2 - F_1'(F_1'+\sqrt{2}) + 4} = \frac{-18(7-2\sqrt{2}F_2')}{[F_2'^2 - F_1'(\sqrt{2}+F_1') + 4]^2}F_2'.
\end{align}
Thus the minimum value of the logarithmic term can only be found on the following boundary of the input fidelity regions: $F_1=1/2$ and $F_2=1/2$. By evaluating the value of the argument on the boundary it can be easily shown that the logarithmic term is always positive on the entire input fidelity region. In this way, we have shown that the factor before $F_2'$ on the right hand side (RHS) of Eqn.~\ref{eqn:dist_ent_F2'deriv} is always non-positive, which means that $E_{D}^\mathrm{low-up}$ decreases when the difference between $F_1$ and $F_2$ increases for all fixed $F_1+F_2$. Therefore, we know that for all square region of input fidelities $(F_1,F_2)\in[F,1]\times[F,1]$ where $F\geq 1/2$ the minimum value of $E_{D}^\mathrm{low-up}$ can only be taken on the boundary $F_1=F,\ F_2=F,\ F_1=1,\ F_2=1$. Hence, suppose there exists a threshold for input fidelity $F_\mathrm{th}$ s.t. $E_{D}^\mathrm{low-up}\geq\forall(F_1,F_2)\in[F_\mathrm{th},1]\times[F_\mathrm{th},1]$, $E_{D}^\mathrm{low-up}$ must be non-negative on all four boundaries $F_1=F_\mathrm{th},\ F_2=F_\mathrm{th},\ F_1=1,\ F_2=1$. According to the symmetry between $F_1$ and $F_2$, we can focus on only two boundaries, e.g. $F_1=F_\mathrm{th},\ F_2=1$.

Firstly, we examine the boundary $F_2=1$, where the objective function $E_{D}^\mathrm{low-up}$ becomes
\begin{equation}
\begin{aligned}
    E_{D}^\mathrm{low-up}(F_1,1) =& \frac{3F_1}{2F_1+1}\log\left(\frac{3F_1}{2F_1+1}\right) + \left(1-\frac{3F_1}{2F_1+1}\right)\log\left[\frac{1}{3}\left(1-\frac{3F_1}{2F_1+1}\right)\right]\\
    &- \frac{1-F_1}{2}\log\left(\frac{1-F_1}{2}\right) - \frac{1+F_1}{2}\log\left(\frac{1+F_1}{2}\right),
\end{aligned}
\end{equation}
for which we can examine its convexity/concavity by taking the second-order derivative as 
\begin{equation}
    \frac{d^2}{dF_1^2}E_{D}^\mathrm{low-up}(F_1,1) = \frac{-3 - 8 F_1 + 12 F_1^3 + 8 F_1^4 + 
 12 F_1 (F_1^2 - 1) \log\left(\frac{1 - F_1}{9 F_1}\right)}{F_1 (1 + 2 F_1)^3 (F_1^2 - 1)}.
\end{equation}    
It is obvious that the denominator is always non-positive on $F_1\in[1/2,1]$, and we now only focus on the numerator $N(F_1)$. As $\frac{1-F_1}{9F_1}<\frac{1}{9F_1}<1$ for $F_1\in[1/2,1]$, we have $\log\left(\frac{1-F_1}{9F_1}\right)<\log\left(\frac{1}{9F_1}\right)<0$. Also given that $F_1^2 - 1\leq 0$ for $F_1\in[1/2,1],$ we have $(F_1^2 - 1)\log\left(\frac{1-F_1}{9F_1}\right) > (F_1^2 - 1)\log\left(\frac{1}{9F_1}\right)$. Now we switch to a relaxed version of the numerator
\begin{equation}
    \tilde{N}(F_1)= -3 - 8 F_1 + 12 F_1^3 + 8 F_1^4 + 12 F_1 (1 - F_1^2) \log\left(9 F_1\right),
\end{equation}
s.t. $\tilde{N}(F_1)\leq N(F_1)$ for $F_1\in[1/2,1]$.
Through the evaluation of its derivative we arrive at
\begin{equation}
    \frac{d}{dF_1}\tilde{N}(F_1) = 4 [1 + 6 F_1^2 + 8 F_1^3 + (3 - 9 F_1^2) \log(9 F_1)],
\end{equation}
within which the only term that can be negative on $[1/2,1]$ is $(3 - 9 F_1^2)$. We thus consider a relaxation $\log(9 F_1)\rightarrow \log(9)\geq\log(9 F_1)>0$. After replacing $\log(9 F_1)$, the new function will still be non-negative when $3 - 9 F_1^2\geq 0$ and will be smaller than the original function when $3 - 9 F_1^2 < 0$. Also, after the transformation, the funtion becomes a simple polynomial which can be straightforwardly shown to be positive on $[1/2,1]$. In this way, we have shown that $\tilde{N}(F_1)$ monotonically increases on $[1/2,1]$, which means that it finds its minimum value on the interval at $F_1=1/2$, which is positive. Therefore, we now know that the numerator of $E_{D}^\mathrm{low-up}(F_1,1)$ is positive when its denominator is non-positive, i.e. $E_{D}^\mathrm{low-up}(F_1,1)$ is concave on $[1/2,1]$, and it can thus have at most two zeros on this interval. 

It is then easy to show that $E_{D}^\mathrm{low-up}(1,1)=0$, and we can numerically find the only other root on $[1/2,1)$ as $F_\mathrm{1,root}\approx 0.879$, above which $E_{D}^\mathrm{low-up}(F_1,1)\geq 0$ and below which $E_{D}^\mathrm{low-up}(F_1,1)\leq 0$. Now combining this numerically determined value with the above proved result that the $F_2'=0$ ``ridge'' has the highest value given a fixed $F_1'$, we can easily see that for $(F_1,F_2)\in[F_\mathrm{th,1},1]\times[F_\mathrm{th,1},1]$ it is guaranteed that $E_{D}^\mathrm{low-up}(F_1,F_2)\geq 0$, where $F_\mathrm{th,1} = (1+F_\mathrm{1,root})/2\approx 0.939$ according to a simple geometric argument. 
\end{proof}

\subsection{Proof of Proposition~\ref{thm:general_bds_guarantee}}
Similar to previous proofs, we can define a fidelity improvement function as
\begin{equation}
    F_\mathrm{incr,BDS}^\mathrm{avg}(F_1,F_2) = \frac{F_1F_2 +a^2(1-F_1)(1-F_2)}{[F_1+a(1-F_1)][F_2+a(1-F_2)] + (1-a)^2(1-F_1)(1-F_2)} - \frac{F_1+F_2}{2}.
\end{equation}
In the following we will prove the three statements of Proposition~\ref{thm:general_bds_guarantee} separately.

\subsubsection{Proof of statement 2}
We first prove statement 2, i.e. for $a>1/3$ there exists $(F_1,F_2)\in[1/2,1]\times[1/2,1]$ s.t. $F_\mathrm{incr,BDS}^\mathrm{avg}<0$.
\begin{proof}
It is sufficient to find explicit values of $F_1,F_2$ satisfying the above condition to prove this statement. Inspired by previous results, we examine the boundary of input fidelity region. Specifically, we compute 
\begin{equation}
    F_\mathrm{incr,BDS}^\mathrm{avg}(1,1/2) = \frac{1}{1+a} - \frac{3}{4},
\end{equation}
which monotonically decreases on $a\in[0,1]$, and obviously its zero is at $a=1/3$. Therefore, it is clear that for $a>1/3$ we have $F_\mathrm{incr,BDS}^\mathrm{avg}(1,1/2)<0$.
\end{proof}

\subsubsection{Proof of statement 3}
Then we prove statement 3, i.e. $F_\mathrm{incr,BDS}^\mathrm{avg}(F_1,F_2)\leq 0$ for all $(F_1,F_2)\in[1/2,1]\times[1/2,1]$ if $a\geq 1/2$. 
\begin{proof}
Given the symmetry of $F_\mathrm{incr,BDS}^\mathrm{avg}(F_1,F_2)$ w.r.t. $F_1=F_2$, we perform a natural coordinate transformation 
\begin{equation}
 F_1'=\frac{F_1+F_2}{\sqrt{2}}~(\text{diagonal}),~F_2'=\frac{F_2-F_1}{\sqrt{2}}~(\text{anti-diagonal}),
\end{equation}
after which the input fidelity region becomes $F_1'\in[1/\sqrt{2},\sqrt{2}]$, while $F_2'\in[1/\sqrt{2}-F_1',F_1'-1/\sqrt{2}]$ for $F_1'\in[1/\sqrt{2},3/2\sqrt{2}]$ and $F_2'\in[F_1'-\sqrt{2},\sqrt{2}-F_1']$ for $F_1'\in[3/2\sqrt{2},\sqrt{2}]$. The improvement function can be written as
\begin{equation}
    F_\mathrm{incr,BDS}^\mathrm{avg}(F_1',F_2') = \frac{\sqrt{2}(1-a)^2F_1'^3 - (5a^2-6a+3)F_1'^2 + \sqrt{2}(4a^2-2a+1)F_1' + [(1+a^2) + \sqrt{2}(1-a)^2F_1']F_2'^2 - 2a^2}{2[(1-a)^2F_2'^2 - (1-a)^2F_1'^2 + \sqrt{2}(2a^2 - 3a + 1)F_1' - (2a^2 - 2a + 1)]}.
\end{equation}

We then take the partial derivatives w.r.t. $F_2'$ (anti-diagonal)
\begin{equation}
    \frac{\partial}{\partial F_2'}F_\mathrm{incr,BDS}^\mathrm{avg}(F_1',F_2') = 
    \frac{\sqrt{2}(a^3+a^2-3a+1)F_1' - (2a^3+a^2-2a+1)}{[(1-a)^2F_2'^2 - (1-a)^2F_1'^2 + \sqrt{2}(2a^2 - 3a + 1)F_1' - (2a^2 - 2a + 1)]^2}F_2'
\end{equation}    
The denominator of the RHS is always non-negative. Thus for a fixed $F_1'$ the sign of the partial derivative w.r.t. $F_2'$ is determined by the numerator alone: if the numerator is positive the derivative is negative (positive) for negative (positive) $F_2'$, and if the numerator is negative the derivative is positive (negative) for negative (positive) $F_2'$. In other words, for a fixed $F_1'$ the minimum (maximum) of $F_\mathrm{incr,BDS}^\mathrm{avg}(F_1',F_2')$ is obtained at $(F_1',0)$ if the numerator is positive (negative). Therefore, we examine the numerator as a function of $F_1'$ and $a$
\begin{equation}
    N_{\partial F_2'}(F_1',a) = \sqrt{2}(a^3+a^2-3a+1)F_1' - (2a^3+a^2-2a+1).
\end{equation}

As a routine, we take the partial derivative w.r.t. $a$ of $N_{\partial F_2'}$ and obtain
\begin{equation}
    \frac{\partial}{\partial a}N_{\partial F_2'}(F_1',a) = \sqrt{2} (3 a^2 + 2a - 3) F_1' - (6a^2 + 2a - 2).
\end{equation}
For $a\in[(\sqrt{10}-1)/3,1]$ the partial derivative is upper bounded by $2(3 a^2 + 2a - 3) - (6a^2 + 2a - 2) = -3a^2 - 1$, and for $a\in[0,(\sqrt{10}-1)/3]$ it is upper bounded by $(3 a^2 + 2a - 3) - (6a^2 + 2a - 2) = 2a - 4$. Notably these two upper bounds on $a\in[0,1]$ are both below 0, so $\partial N_{\partial F_2'}/\partial a<0$ for all $F_1'\in[1/\sqrt{2},\sqrt{2}]$ and $a\in[0,1]$. We also observe the following factorization $a^3+a^2-3a+1 = (a-1)(a+1-\sqrt{2})(a+1+\sqrt{2})$. Therefore, $a^3+a^2-3a+1>0$ for $a\in[0,\sqrt{2}-1)$ and $a^3+a^2-3a+1<0$ for $a\in(\sqrt{2}-1,1)$, while $a^3+a^2-3a+1=0$ at $a=\sqrt{2}-1$ (for which $N_{\partial F_2'}\approx -0.485<0$) or $a=1$ (for which $N_{\partial F_2'}=-2<0$). Now we can already see that $N_{\partial F_2'}(F_1',a)<0$ for all $F_1'\in[1/\sqrt{2},\sqrt{2}]$ when $a\in(\sqrt{2}-1,1)$, because under this condition the maximum of $N_{\partial F_2'}(F_1',a)$ must be taken at $F_1'=1/\sqrt{2}$
\begin{align}
    N_{\partial F_2'}(1/\sqrt{2},a) = (a^3+a^2-3a+1) - (2a^3+a^2-2a+1) = -a(a^2+1) \leq 0.
\end{align}

Then we focus on $a\in[0,\sqrt{2}-1)$ so that $a^3+a^2-3a+1 > 0$, and solve $N_{\partial F_2'}(F_1', a)=0$ for $(F_1',a)\in[1/\sqrt{2},\sqrt{2}]\times[0,1]$. The equation is equivalent to 
\begin{equation}
    F_1' = \frac{2a^3+a^2-2a+1}{\sqrt{2}(a^3+a^2-3a+1)},
    \label{eqn:f1'_vs_a}
\end{equation}
We further take the derivative of the RHS of Eqn.~\ref{eqn:f1'_vs_a}
\begin{equation}
    \frac{d}{da}F_1'(a) = \frac{1 + 2a^2 - 8a^3 + a^4}{\sqrt{2}(a^3+a^2-3a+1)^2},
\end{equation}
where $F_1'(a)$ denotes the RHS of Eqn.~\ref{eqn:f1'_vs_a}. The sign of the derivative is only determined by the numerator $(1 + 2a^2 - 8a^3 + a^4)$. The derivative of the numerator w.r.t. $a$ is $4a(a - 3 - 2 \sqrt{2})(a - 3 + 2\sqrt{2})$, which first increases on $[0,3-2\sqrt{2}]$ and then decreases on $[3-2\sqrt{2},\sqrt{2}-1)$. Therefore, the minimum of the numerator on $[0,\sqrt{2}-1)$ can only be obtained at $a=0$ or $a=\sqrt{2}-1$. We then see that the numerator is always positive on $[0,\sqrt{2}-1)$, hence the derivative of $F_1'(a)$ is always positive on $[0,\sqrt{2}-1)$ as well. The minimum of $F_1'(a)$ is then $F_1'(0)=1/\sqrt{2}$. In other words, $N_{\partial F_2'}(1/\sqrt{2},0)=0$.

Then we want to know whether there is $F_1'$ that satisfies Eqn.~\ref{eqn:f1'_vs_a} for $a>0$. As $F_1'(a)$ tends to $\infty$ when $a$ approaches $\sqrt{2}-1$, and $F_1'(a)$ is a monotonic elementary function of $a$ on $[0,\sqrt{2}-1)$, there must exist one and only one $a$ s.t. $F_1'(a)=\sqrt{2}$. In fact, we can directly solve the following equation
\begin{equation}
    \frac{2a^3 + a^2 - 2a + 1}{\sqrt{2}(a^3+a^2-3a+1)} = \sqrt{2} \Rightarrow a = 2\pm\sqrt{3}.
\end{equation}
Only one root $a=2-\sqrt{3}$ is within $[0,\sqrt{2}-1)$. For $a>2-\sqrt{3}$, we will have $F_1'(a)>\sqrt{2}$ so there is no valid $F_1'$ that satisfies Eqn.~\ref{eqn:f1'_vs_a}. On the other hand, and for all $F_1'\in[1/\sqrt{2},\sqrt{2}]$ there exists one and only one $a(F_1')\in[0,2-\sqrt{3}]$ as determined by inverse of Eqn.~\ref{eqn:f1'_vs_a}, s.t. $N_{\partial F_2'}(F_1', a(F_1'))=0$. Note that $a(F_1')$ is also monotonic on $F_1'\in[1/\sqrt{2},\sqrt{2}]$.

In summary, for all $F_1'\in[1/\sqrt{2},\sqrt{2}]$ we have $N_{\partial F_2'}(F_1', a)>0$ when $a<a(F_1')\leq 2-\sqrt{3}$, and $N_{\partial F_2'}(F_1', a)<0$ for all $a\in(2-\sqrt{3},1]$.

Now we evaluate the maximum of $F_\mathrm{incr,BDS}^\mathrm{avg}(F_1',F_2')$. As $1/2>2-\sqrt{3}$ for all $a\geq 1/2$ we have $N_{\partial F_2'}(F_1', a)<0$, therefore the maximum of $F_\mathrm{incr,BDS}^\mathrm{avg}(F_1',F_2')$ is obtained at $F_2'=0$ for a certain $F_1'\in[1/\sqrt{2},\sqrt{2}]$. Hence we only need to examine the maximum of $F_\mathrm{incr,BDS}^\mathrm{avg}(F_1',0)$ on $F_1'\in[1/\sqrt{2},\sqrt{2}]$. Thus we can define a new function
\begin{equation}
    \tilde{F}_\mathrm{incr,BDS}^\mathrm{avg}(F_1',a) = F_\mathrm{incr,BDS}^\mathrm{avg}(F_1',0) = \frac{-2 a^2 + \sqrt{2} (1 - 2 a + 4 a^2) F_1' - (3 - 6a + 5 a^2) F_1'^2 + \sqrt{2} (1 - a)^2 F_1'^3}{2 [-1 + 2a - 2a^2 - \sqrt{2}(1 - a) (2 a - 1) F_1' - (1 - a)^2 F_1'^2]}.
\label{eqn:ridge}
\end{equation}
Our objective is to prove $\tilde{F}_\mathrm{incr,BDS}^\mathrm{avg}(F_1',a)\leq 0$ for $F_1'\in[1/\sqrt{2},\sqrt{2}], a\in[1/2,1]$. We first examine the partial derivative of $\tilde{F}_\mathrm{incr,BDS}^\mathrm{avg}(F_1',a)$ w.r.t. $a$
\begin{equation}
    \frac{\partial}{\partial a}\tilde{F}_\mathrm{incr,BDS}^\mathrm{avg} = \frac{4 (a - a^2) + 2 \sqrt{2} (5 a^2 - 4a) F_1' + 2 (1 + 5a - 9 a^2) F_1'^2 + \sqrt{2} (7 a^2 - 2a -3) F_1'^3 - 2 (a^2 - 1) F_1'^4}{2 [-1 + 2a - 2a^2 - \sqrt{2}(1 - a) (2 a - 1) F_1' - (1 - a)^2 F_1'^2]^2},
\label{eqn:ridge_a_deriv}
\end{equation}
whose sign is only determined by the numerator. 

We can rewrite the numerator as
\begin{equation}
    \tilde{N}(F_1',a) = (2 F_1'^2 - 3 \sqrt{2} F_1'^3 + 2 F_1'^4) + a (4 - 8 \sqrt{2} F_1' + 10 F_1'^2 - 2 \sqrt{2} F_1'^3) + a^2 (-4 + 10 \sqrt{2} F_1' - 18 F_1'^2 + 7 \sqrt{2} F_1'^3 - 2 F_1'^4).
\label{eqn:ridge_a_deriv_numer}
\end{equation}    
The $F_1'$-polynomial coefficient of the $a^2$ term can be proved to be non-negative on $F_1'\in[1/\sqrt{2},\sqrt{2}]$, i.e. for all fixed $F_1'\in[1/\sqrt{2},\sqrt{2}]$ the numerator $\tilde{N}(F_1',a)$ will either be a parabola opening upward or a straight line (note that at $F_1'=1/\sqrt{2},~\sqrt{2}$ the coefficients of both $a^2$ and $a$ are simultaneously zero, thus the straight line corresponds to a constant $\tilde{N}=0$). Now we are interested in evaluating the sign of the numerator, which can be determined by checking the two end points $\tilde{N}(F_1',1/2)$ and $\tilde{N}(F_1',1)$ for all fixed $F_1'$
\begin{equation}
    \tilde{N}(F_1',1/2) = \frac{1}{4}(2 - 3 \sqrt{2} F_1' + 2 F_1'^2) (2 + 3 F_1'^2),~\tilde{N}(F_1',1) = 2 F_1' (\sqrt{2} - 3 F_1' + \sqrt{2} F_1'^2).
\end{equation}
We see that $\tilde{N}(F_1',1/2)\leq 0$ and $\tilde{N}(F_1',1)\leq 0$ for $F_1'\in[1/\sqrt{2},\sqrt{2}]$. Combining previous arguments, we are now sure that $\tilde{F}_\mathrm{incr,BDS}^\mathrm{avg}$ monotonically decreases as $a$ increases. Therefore, the maximum of $\tilde{F}_\mathrm{incr,BDS}^\mathrm{avg}$ can only be obtained for the smallest possible $a$. We then evaluate $\tilde{F}_\mathrm{incr,BDS}^\mathrm{avg}(F_1',1/2)$
\begin{equation}
    \tilde{F}_\mathrm{incr,BDS}^\mathrm{avg}(F_1',1/2) = \frac{5}{2} - \frac{F_1'}{\sqrt{2}} - \frac{4 + \sqrt{2}F_1'}{2 + F_1'^2},
\end{equation}
which is indeed non-positive for $F_1'\in[1/\sqrt{2},\sqrt{2}]$. Hence statement 3 is proved.
\end{proof}

\begin{remark}
For all $a<1/2$ there always exists $(F_1,F_2)\in[1/2,1]\times[1/2,1]$ that achieves positive fidelity improvement w.r.t. $(F_1+F_2)/2$, according to the partial derivative of $F_\mathrm{incr,BDS}^\mathrm{avg}(F_1',F_2')$ w.r.t. $F_1'$ at $F_1'=\sqrt{2},F_2'=0$
\begin{equation}
    \left.\frac{\partial}{\partial F_1'}F_\mathrm{incr,BDS}^\mathrm{avg}\right\vert_{F_1'=\sqrt{2},F_2'=0} = \frac{2a-1}{\sqrt{2}}.
\end{equation}
When $a<1/2$ the above value is negative. Therefore, there exists $(F_1',F_2')$ in the neighborhood of $(\sqrt{2},0)$ (and also in the possible region of $(F_1',F_2')$)  s.t. $F_\mathrm{incr,BDS}^\mathrm{avg}(F_1',F_2')>0$.
\end{remark}

\subsubsection{Proof of statement 1}
Lastly, we prove statement 1, i.e. $F_\mathrm{incr,BDS}^\mathrm{avg}\geq 0$ for all $F_1,F_2\in[1/2,1]$, if $a\leq 1/3$.
\begin{proof}
According to previous results, we separate two different cases, i.e. $a\in[0,2-\sqrt{3}]$ and $a\in[2-\sqrt{3},1/3]$. For $a\in[2-\sqrt{3},1/3]$, we know that $N_{\partial F_2'}(F_1', a)<0$, i.e. $\frac{\partial}{\partial F_2'}F_\mathrm{incr,BDS}^\mathrm{avg} \geq(\leq) 0$ when $F_2'\leq(\geq)0$. Hence, for all fixed $F_1'$, the minimum of $F_\mathrm{incr,BDS}^\mathrm{avg}$ is always obtained on the boundary of possible region of $F_2'$ which corresponds to the boundary of the original $(F_1,F_2)$ region. For $a\in[0,2-\sqrt{3}]$, we know that $N_{\partial F_2'}(F_1', a)$ increases monotonically as $a$ increases. Hence $N_{\partial F_2'}(F_1', a)$ is first negative and then positive as $a$ increases. In other words, for a fixed $F_1'$, $F_\mathrm{incr,BDS}^\mathrm{avg}(F_1',0)$ is maximum when $F_1'<F_1'(a)$ (expression given in Eqn.~\ref{eqn:f1'_vs_a}), and is minimum when $F_1'>F_1'(a)$. Hence, the minimum of $F_\mathrm{incr,BDS}^\mathrm{avg}$ can only be obtained on the boundary of the $(F_1,F_2)$ region, or on the ridge at $F_2'=0\ (F_1=F_2)$. Therefore, we examine the boundary $F_2=1/2$ and $F_1=1$ (as $F_\mathrm{incr,BDS}^\mathrm{avg}(F_1,F_2)$ is symmetric against $F_1=F_2$), and the ridge $F_2'=0\ (F_1=F_2)$. 

On the boundary $F_2=1/2$ we have
\begin{equation}
    F_\mathrm{incr,BDS}^\mathrm{avg}(F_1,1/2) = \frac{(2F_1-1)[(1+a)(1-2a) - 2a(1-a)F_1]}{4[(1-a+2a^2) + 2a(1-a)F_1]}.
\end{equation}
Its partial derivative w.r.t. $a$ is
\begin{equation}
    \frac{\partial}{\partial a}F_\mathrm{incr,BDS}^\mathrm{avg}(F_1,1/2) = -\frac{[(1-a)^2F_1 + a(2-a)](2F_1 - 1)}{[(1-a+2a^2) + 2a(1-a)F_1]^2} \leq 0.
\end{equation}
Hence the minimum of $F_\mathrm{incr,BDS}^\mathrm{avg}(F_1,1/2)$ can only be obtained at $a=2-\sqrt{3}$, which gives
\begin{equation}
    \left.F_\mathrm{incr,BDS}^\mathrm{avg}(F_1,1/2)\right\vert_{a=2-\sqrt{3}} = \frac{-(3\sqrt{3}-5)(4F_1^2 - 8F_1 + 3)}{(52-28\sqrt{3}) + 8(3\sqrt{3}-5)F_1}.
\end{equation}
The above function is concave, so its minimum can only be obtained at the boundary of the $F_1$ domain, i.e.
\begin{equation}
    \left.F_\mathrm{incr,BDS}^\mathrm{avg}(1/2,1/2)\right\vert_{a=2-\sqrt{3}} = 0,~\left.F_\mathrm{incr,BDS}^\mathrm{avg}(1,1/2)\right\vert_{a=2-\sqrt{3}} = \frac{3\sqrt{3}-5}{12-4\sqrt{3}} > 0.
\end{equation}
Thus, $F_\mathrm{incr,BDS}^\mathrm{avg}(F_1,F_2)$ will always be non-negative on the boundary $F_2=1/2$. 

On the boundary $F_1=1$ we have
\begin{equation}
    F_\mathrm{incr,BDS}^\mathrm{avg}(1,F_2) = \frac{F_2}{a(1 - F_2) + F_2} - \frac{1+F_2}{2}.
\end{equation}
Again, its $a$ partial derivative is
\begin{equation}
    \frac{\partial}{\partial a}F_\mathrm{incr,BDS}^\mathrm{avg}(1,F_2) = -\frac{(1-F_2)F_2}{[a(1 - F_2) + F_2]^2} \leq 0.
\end{equation}
Hence the minimum of $F_\mathrm{incr,BDS}^\mathrm{avg}(1,F_2)$ can only be obtained at $a=2-\sqrt{3}$, and with this choice of $a$ we have
\begin{equation}
    \left.F_\mathrm{incr,BDS}^\mathrm{avg}(1,F_2)\right\vert_{a=2-\sqrt{3}} =\frac{F_2}{(2-\sqrt{3})(1 - F_2) + F_2} - \frac{1+F_2}{2}.
\end{equation}
The above function is concave, so its minimum can only be obtained at the boundary of the $F_2$ domain, i.e.
\begin{equation}
    \left.F_\mathrm{incr,BDS}^\mathrm{avg}(1,1/2)\right\vert_{a=2-\sqrt{3}} = \frac{1}{3-\sqrt{3}} - \frac{3}{4} > 0,~\left.F_\mathrm{incr,BDS}^\mathrm{avg}(1,1)\right\vert_{a=2-\sqrt{3}} = 0.
\end{equation}
Thus, $F_\mathrm{incr,BDS}^\mathrm{avg}(F_1,F_2)$ will always be non-negative on the boundary $F_1=1$. 

On the ridge $F_2'=0$, the relevant functions have been obtained, including Eqn.~\ref{eqn:ridge}, and its partial derivative w.r.t. $a$ Eqn.~\ref{eqn:ridge_a_deriv}, whose numerator is in Eqn.~\ref{eqn:ridge_a_deriv_numer}. We are interested in the sign of the numerator, and we first check the two end points $\tilde{N}(F_1',0)$ and $\tilde{N}(F_1',1/3)$
\begin{equation}
    \tilde{N}(F_1',0) = 2F_1'^2 - 3\sqrt{2}F_1'^3 + 2F_1'^4,~\tilde{N}(F_1',1/3) = \frac{2}{9}(4 - 7 \sqrt{2} F_1' + 15 F_1'^2 - 13 \sqrt{2} F_1'^3 + 8 F_1'^4).
\end{equation}
We can see that $\tilde{N}(F_1',0)\leq 0$ and $\tilde{N}(F_1',1/3)\leq 0$ for $F_1'\in[1/\sqrt{2},\sqrt{2}]$, so $\tilde{F}_\mathrm{incr,BDS}^\mathrm{avg}$ monotonically decreases as $a$ increases on $[0,2-\sqrt{3}]$, according to the upward-opening parabola interpretation of the numerator function $\tilde{N}(F_1',a)$. Therefore, the minimum of $\tilde{F}_\mathrm{incr,BDS}^\mathrm{avg}$ can only be obtained at $a=2-\sqrt{3}$:
\begin{equation}
    \tilde{F}_\mathrm{incr,BDS}^\mathrm{avg}(F_1',2-\sqrt{3}) = 1 - \frac{F_1'}{\sqrt{2}} + \frac{\sqrt{2}F_1' - 2}{2F_1'^2 + (\sqrt{6}-3\sqrt{2})F_1' + 4 - \sqrt{3}},
\end{equation}
which can be shown to be non-negative for $F_1'\in[1/\sqrt{2},\sqrt{2}]$. 

Combining all three above cases, the statement is proved. 
\end{proof}

\section{Additional comments for guaranteed improvement}
\subsection{Comment on probabilistic success}
EPPs in general succeed with probability below 1. We have focused on the figures of merit conditioned on success, and compare them with the properties of input states which are treated as baselines to prove guaranteed improvement. In fact, such post-selection is extremely important for probabilistic EPPs in general. Specifically, we can show that for this protocol as an example, the successful output fidelity normalized by the success probability will never be higher than the lower input fidelity.
\begin{proposition}
    for all pair of BDS's with fidelity above 1/2 as inputs to the recurrence EPP, the normalized output fidelity is not higher than the lower input fidelity.
\end{proposition}
\begin{proof}
    According to Eqn.~\ref{eqn:bds_elems_dejmps}, for the first BDS characterized by diagonal elements $\vec{\lambda}=(\lambda_1,\lambda_2,\lambda_3,\lambda_4)$ and the second by $\vec{\lambda}'=(\lambda_1',\lambda_2',\lambda_3',\lambda_4')$, the normalized output fidelity is $\tilde{F}=\lambda_1\lambda_1' + \lambda_2\lambda_2'=F_1F_2 + \lambda_2\lambda_2'$, where we use tilde to denote quantity normalized by success probability, and we have renamed the first element of the diagonal-element vector by fidelity. Then consider the following difference
    \begin{equation}
        \tilde{F}-F_1 = \lambda_2\lambda_2' - (1-F_2)F_1 \leq (1-F_1)(1-F_2) - (1-F_2)F_1 = (1-F_2)(1-2F_1) \leq 0.
    \end{equation}
    Similarly we can also show that $\tilde{F}-F_2\leq 0$. Therefore, we have proved that $\tilde{F}\leq\min\{F_1,F_2\}$.
\end{proof}

\subsection{Other guaranteed improvement}
In the main text, we were focused on entanglement measures of average input state (AIS) as the baseline to evaluate the improvement from successful EPP for Werner input states. Here, we present another guaranteed improvement with different baselines, i.e. logarithmic negativity w.r.t. the average of the two input states' logarithmic negativities, for Werner state.
\begin{proposition}\label{thm:avg_ln_improv}
    Let $F^{(W)}(F_1,F_2)$ be the successful output fidelity of the recurrence EPP given two input Werner states $\rho_\mathrm{W}(F_1)$ and $\rho_\mathrm{W}(F_2)$ with $F_1,F_2\geq 1/2$, then the logarithmic negativity of output state is $E_N(F^{(W)}(F_1,F_2))$, and we have $E_N(F^{(W)}(F_1,F_2))\geq [E_N(F_1)+E_N(F_2)]/2$.
\end{proposition}
\begin{proof}
   We define the increase function as
    \begin{equation}
        E_{N,\mathrm{incr}}^{(W)} = E_N(F^{(W)}(F_1,F_2)) - \frac{E_N(F_1) + E_N(F_2)}{2}.
    \end{equation}
    According to Eqn.~\ref{eqn:log_neg} we have $\frac{d^2}{dF^2}E_N(F)\propto -\frac{1}{F^2} < 0$
    which means that the logarithmic negativity as a function of BDS fidelity is concave. Given the concavity, we have that
    \begin{equation}
        \frac{E_N(F_1) + E_N(F_2)}{2} \leq E_N(\frac{F_1+F_2}{2}).
    \end{equation}
    Lastly recall Prop.~\ref{thm:werner_fidincr} and we complete the proof.
\end{proof}
This improvement is due to the underlying entanglement measure function's concavity. It can be shown that the convex upper and lower bounds of distillable entanglement considered in this work are not guaranteed to be improved w.r.t. the average of the input states' upper and lower bounds.

\section{Proofs for optimal time}\label{sec:proof_time}
In this appendix, we provide detailed proofs of the propositions on optimal time in Sec~\ref{sec:opt_time}. 

\subsection{Proof of Proposition~\ref{thm:rank2_fid_opt_time}}
\begin{proof}
    Firstly, we have the fidelity dynamics of the bit-flipped Bell state under memory bit-flip channel as
    \begin{equation}
        F_\mathrm{bit-flip}(F_0, t) = F_0e^{-2\kappa t} + \frac{1-e^{-2\kappa t}}{2},
    \end{equation}
    where $F_0$ is the raw fidelity. According to the scenario, we can give an analytical expression of the fidelity at $t=t_2$ upon successful purification
    \begin{equation}
        F_\mathrm{succ}^{(2)}(t_2) = F_\mathrm{bit-flip}(F_\mathrm{succ}^{(2)}(t),t_2-t),
    \end{equation}
    where we use superscript (2) to denote the rank-2 BDS, and we have
    \begin{equation}
    \begin{aligned}
        F_\mathrm{succ}^{(2)}(t) =& \frac{F_\mathrm{bit-flip}(F_0, t)F_\mathrm{bit-flip}(F_0, t-t_1)}{F_\mathrm{bit-flip}(F_0, t)F_\mathrm{bit-flip}(F_0, t-t_1) + [1-F_\mathrm{bit-flip}(F_0, t)][1-F_\mathrm{bit-flip}(F_0, t-t_1)]}\\
        =& \frac{1 - (1-2F_0)[e^{-2\kappa t} + e^{-2\kappa(t-t_1)}] + (1-2F_0)^2e^{-2\kappa(2t-t_1)}}{2 + 2(1-2F_0)^2e^{-2\kappa(2t-t_1)}},
    \end{aligned}
    \end{equation}
    and thus after successful purification the fidelity at $t_2$ is
    \begin{equation}
        F_\mathrm{succ}^{(2)}(t_2) = \frac{1 - (1-2F_0)[e^{-2\kappa t_2} + e^{-2\kappa(t_2-t_1)}] + (1-2F_0)^2e^{-2\kappa(2t-t_1)}}{2 + 2(1-2F_0)^2e^{-2\kappa(2t-t_1)}}.
    \end{equation}
    Then it is straightforward to evaluate the derivative of $[F_\mathrm{succ}^{(2)}(t_2)](t)$ w.r.t. $t$, whose analytical expression is
    \begin{equation}\label{eqn:fid_derivative}
        \frac{d}{dt}[F_\mathrm{succ}^{(2)}(t_2)](t) = \frac{2\kappa(2F_0-1)^3(1+e^{2\kappa t_1})e^{2\kappa(2t+t_1-t_2)}}{[e^{4\kappa t} + (1-2F_0)^2e^{2\kappa t_1}]^2}.
    \end{equation}    
    Given that meaningful entanglement generation should have $F_0>1/2$, we should always have $2F_0-1>0$, and this means that $\frac{d}{dt}[F_\mathrm{succ}^{(2)}(t_2)](t)$ is always positive, i.e. $F_\mathrm{succ}^{(2)}(t_2)$ increases monotonically with increasing $t$. 
\end{proof}

\subsection{Proof of Proposition~\ref{thm:rank2_normfid_opt_time}}
\begin{proof}
    We first prove the statement for fidelity. The explicit analytical expression for normalized fidelity at $t_2$ is
    \begin{equation}
        \tilde{F}_\mathrm{succ}^{(2)}(t_2) = \frac{1 - (1-2F_0)(e^{-2\kappa t_2} + e^{-2\kappa(t_2-t_1)}) + (1-2F_0)^2e^{-2\kappa(2t-t_1)}}{4}.
    \end{equation}
    And the derivative w.r.t. $t$ is
    \begin{equation}
        \frac{d}{dt}[\tilde{F}_\mathrm{succ}^{(2)}(t_2)](t) = -\kappa(1-2F_0)^2e^{-2\kappa(2t-t_1)} < 0,
        \label{eqn:fid_norm_derivative}
    \end{equation}
    which means that the later we perform the EPP, the lower normalized fidelity we can get due to decrease in the EPP success probability over time.

    Next we move on to concurrence. Note that although both concurrence and negativity (which is half the concurrence) are linear in fidelity ($\propto 2F-1$), under the definition of normalized quantity the optimal time for normalized fidelity does not naturally apply to normalized concurrence (negativity) because $-p_\mathrm{succ}(t)$ increases when $t$ increases. The expression of normalized concurrence at $t_2$ is 
    \begin{equation}
        \tilde{\mathcal{C}}_\mathrm{succ}^{(2)}(t_2) = 1 + \frac{(2F_0-1)\left[e^{-2\kappa t_2}+e^{-2\kappa(t_2-t_1)}+2e^{-2\kappa(2t-t_1)}(2F_0-1)\right]}{2}.
    \end{equation}
    Then it is straightforward to evaluate the derivative of $\tilde{\mathcal{C}}_\mathrm{succ}^{(2)}(t_2)$ w.r.t. $t$ as 
    \begin{equation}
        \frac{d}{dt}[\tilde{\mathcal{C}}_\mathrm{succ}^{(2)}(t_2)](t) = -4\kappa(2F_0-1)^2e^{-2\kappa(2t-t_1)} < 0.
    \end{equation}
    This monotonicity result concludes the proof.
\end{proof}

\subsection{Proof of Proposition~\ref{thm:norm_DELN_rank2_opt_time}}
\begin{proof}
    We first prove the statement for normalized distillable entanglement. The analytical expression of normalized distillable entanglement at $t_2$ can be written as 
    \begin{equation}
    \begin{aligned}
        \tilde{E}_D =& E_D(F'(t_2))[F(F_0, t)F(F_0, t-t_1) + (1-F(F_0, t))(1-F(F_0, t-t_1))]\\
        =& \frac{e^{-4\kappa t}}{4\ln 2}\left[\ln 2A + (B+C)\ln\left(\frac{B+C}{A}\right) + (B-C)\ln\left(1-\frac{B+C}{A}\right)\right],
    \end{aligned}
    \end{equation}        
    where $E_D$ is the distillable entanglement for a bit-flipped Bell state, and we have defined the following quantities
    \begin{equation}
        A = 2\left(e^{4\kappa t} + (2F_0-1)^2e^{2\kappa t_1}\right),\ B = e^{4\kappa t} + (2F_0-1)^2e^{2\kappa t_1},\ C = (2F_0-1)\left(e^{4\kappa t - 2\kappa t_2} + e^{4\kappa t + 2\kappa t_1 - 2\kappa t_2}\right).
    \end{equation}
    Then it is also straightforward to get its derivative w.r.t. $t$
    \begin{equation}
        \frac{d}{dt}\tilde{E}_D = -\frac{\kappa(2F_0-1)^2}{\ln 2}e^{-2\kappa(2t-t_1)}\ln\left[\frac{4(B+C)}{A}\left(1-\frac{B+C}{A}\right)\right].
    \end{equation}        
    It is obvious that the factor before the logarithmic function is always negative, so we examine the logarithmic function itself to understand the behavior of normalized distillable entanglement w.r.t. the EPP time $t$. Explicitly,
    \begin{equation}
    \begin{aligned}
        &\frac{4(B+C)}{A}\left(1-\frac{B+C}{A}\right)\\
        =& \frac{e^{-4\kappa t_2}\left[e^{4\kappa(2t+t_2)} + (2F_0-1)^2\left(2e^{2\kappa(2t+t_1+2t_2)}-2e^{2\kappa(4t+t_1)}-e^{4\kappa(2t+t_1)}-e^{8\kappa t}\right) + (2F_0-1)^4e^{4\kappa(t_1+t_2)}\right]}{e^{8\kappa t} + 2(2F_0-1)^2e^{2\kappa(2t+t_1)} + (2F_0-1)^4e^{4\kappa t_1}}.
    \end{aligned}
    \end{equation}
    Let $u= e^{4\kappa t}$, and we would like to evaluate if the above expression is above or below 1, so we first solve $u$ which makes the above expression equal to 1, i.e.
    \begin{equation}
    \begin{aligned}
        & u^2 + (2F_0-1)^2\left(2e^{2\kappa t_1}u-2e^{2\kappa(t_1-2t_2)}u^2-e^{4\kappa(t_1-t_2)}u^2-e^{-4\kappa t_2}u^2\right) + (2F_0-1)^4e^{4\kappa t_1}\\
        &= u^2 + 2(2F_0-1)^2e^{4\kappa t_1}u + (2F_0-1)^4e^{4\kappa t_1}\\
        \Rightarrow& \left(2e^{2\kappa(t_1-2t_2)} + e^{4\kappa(t_1-t_2)} + e^{-4\kappa t_2}\right)u^2 = \left(1 + e^{2\kappa t_1}\right)^2e^{-4\kappa t_2}u^2 = 0.
        \label{eqn:root_eqn}
    \end{aligned}
    \end{equation}
    The above equation requires $u=e^{4\kappa t}=0$, which does not have a solution with a positive $t$, thus the expression within logarithmic function will never be 1. According to Eqn.~\ref{eqn:root_eqn} we can see that the numerator is always smaller than the denominator, which suggests that the expression itself is smaller than 1, giving the logarithmic function a negative value. Therefore, we now know that the entire time derivative of the normalized distillable entanglement is always positive, and this leads to the interpretation that if we want to obtain the highest normalized distillable entanglement by the end time, we may prefer to perform the EPP just before the end time. Note, however, an interesting, and apparently contradictory, comparison with Eqn.~\ref{eqn:fid_norm_derivative}, which suggests that if the metric of interest is the normalized fidelity it is desirable to perform the EPP by the time the second pair is generated.

    Then we prove the statement for normalized logarithmic negativity. We can again directly obtain the expression of normalized entanglement negativity
    \begin{equation}
        \tilde{E}_N(t) = \frac{1+(2F_0-1)^2e^{-2\kappa(2t-t_1)}}{2}\log_2\left[\frac{1+(2F_0-1)\left(e^{-2\kappa t_2} + e^{-2\kappa(t_2-t_1)}\right) + (2F_0-1)^2e^{-2\kappa(2t-t_1)}}{1+(2F_0-1)^2e^{-2\kappa(2t-t_1)}}\right].
    \end{equation}
    It is also straightforward to obtain its derivative w.r.t. the time for the EPP $t$
    \begin{equation}
    \begin{aligned}
        \frac{d}{dt}\tilde{E}_N =& \frac{2\kappa(2F_0-1)^2e^{-2\kappa(2t-t_1)}}{\ln 2}\left[\frac{(2F_0-1)\left(e^{-2\kappa t_2} + e^{-2\kappa(t_2-t_1)}\right)}{1+(2F_0-1)\left(e^{-2\kappa t_2} + e^{-2\kappa(t_2-t_1)}\right) + (2F_0-1)^2e^{-2\kappa(2t-t_1)}} \right.\\
        &-\left.\ln\left(\frac{1+(2F_0-1)\left(e^{-2\kappa t_2} + e^{-2\kappa(t_2-t_1)}\right) + (2F_0-1)^2e^{-2\kappa(2t-t_1)}}{1+(2F_0-1)^2e^{-2\kappa(2t-t_1)}}\right) \right]\\
        =& \frac{2\kappa(2F_0-1)^2e^{-2\kappa(2t-t_1)}}{\ln 2}\left[\frac{A}{A+B(t)} - \ln\left(\frac{A+B(t)}{B(t)}\right)\right],
        \label{eqn:EN_norm_derivative}
    \end{aligned}
    \end{equation}
    where
    \begin{equation}
        A= (2F_0-1)\left(e^{-2\kappa t_2} + e^{-2\kappa(t_2-t_1)}\right),\ B(t)= 1+(2F_0-1)^2e^{-2\kappa(2t-t_1)}.
    \end{equation}
    By evaluating the $t$ derivative of the term in the square brackets on the last line of Eqn.~\ref{eqn:EN_norm_derivative}, we see that its derivative is always negative. And for the lowest possible $t=t_1$ the bracket term has value
    \begin{equation}
    \begin{aligned}
        \frac{A}{A+B(t_1)} - \ln\left(\frac{A+B(t_1)}{B(t_1)}\right) =& \frac{(2F_0-1)\left(e^{-2\kappa t_2} + e^{-2\kappa(t_2-t_1)}\right)}{(2F_0-1)\left(e^{-2\kappa t_2} + e^{-2\kappa(t_2-t_1)}\right) + 1+(2F_0-1)^2e^{-2\kappa t_1}}\\
        &- \ln\left(\frac{(2F_0-1)\left(e^{-2\kappa t_2} + e^{-2\kappa(t_2-t_1)}\right) + 1+(2F_0-1)^2e^{-2\kappa t_1}}{1+(2F_0-1)^2e^{-2\kappa t_1}}\right).
    \end{aligned}
    \end{equation}
    By further evaluating the above term's partial derivatives w.r.t. $t_1$ and $t_2$, we see that when $t_1$ decreases and when $t_2$ increases its value increases. Therefore, the highest possible value it can take should correspond to the lowest possible $t_1$, i.e. 0, and highest possible $t_2$, i.e. $+\infty$, gives value 0. Then we know that the bracket term is always negative, which leads to negative $t$ derivative of the normalized entanglement negativity. 
\end{proof}

\subsection{Proof of Proposition~\ref{thm:werner_fid_opt_time}}
\begin{proof}
We have the following expression of output state's fidelity at time $t_2$
\begin{equation}
    F'(t_2) = \frac{9 + 3(4F_0-1)(e^{-2\kappa t_2} + e^{-2\kappa(t_2-t_1)}) + (4e^{-2\kappa(t+t_2-t_1)} + e^{-2\kappa(2t-t_1)})(4F_0-1)^2}{36 + 4e^{-2\kappa(2t-t_1)}(4F_0-1)^2},
\end{equation}
and its $t$ derivative is
\begin{equation}
    \frac{d}{dt}\left[F'(t_2)\right](t) = \kappa(4F_0-1)^2\frac{(4F_0-1)^2e^{-2\kappa(2t+t_2-2t_1)} + 3(4F_0-1)(e^{-2\kappa(2t+t_2-t_1)} + e^{-2\kappa(2t+t_2-2t_1)}) - 18e^{-2\kappa(t+t_2-t_1)}}{[9 + e^{-2\kappa(2t-t_1)}(4F_0-1)^2]^2}.
\end{equation}
The root for the time derivative being zero is easy to calculate, according to the equation
\begin{equation}
    \left((4F_0-1)^2e^{-2\kappa(t_2-2t_1)} + 3(4F_0-1)(e^{-2\kappa(t_2-t_1)} + e^{-2\kappa(t_2-2t_1)})\right)\left[e^{-2\kappa t}\right]^2 - 18e^{-2\kappa(t_2-t_1)}\left[e^{-2\kappa t}\right] = 0,
    \label{eqn:derivative_numerator}
\end{equation}
which gives one single finite positive root determined by (the other root corresponds to $e^{-2\kappa t}=0$, i.e. $t\rightarrow\infty$)
\begin{equation}
    e^{-2\kappa t} = \frac{18e^{-2\kappa(t_2-t_1)}}{(4F_0-1)^2e^{-2\kappa(t_2-2t_1)} + 3(4F_0-1)(e^{-2\kappa(t_2-t_1)} + e^{-2\kappa(t_2-2t_1)})}.
\end{equation}    
That is, we have the extremum point of the $t$ derivative
\begin{equation}
    t^* = \frac{1}{2\kappa}\ln\left[\frac{(4F_0-1)^2e^{2\kappa t_1} + 3(4F_0-1)(1 + e^{2\kappa t_1})}{18}\right],
\end{equation}
which is dependent on $t_1$ only. According to Eqn.~\ref{eqn:derivative_numerator} we know that for $0<t<t^*$ the $t$ derivative of $F'(t_2)$ is positive, and for $t>t^*$ the $t$ derivative of $F'(t_2)$ is negative. Therefore $t^*$ is the optimal time to perform the EPP given the error model of depolarizing channel.

We further compare this time with $t_1$ by taking the difference
\begin{equation}
    e^{2\kappa t^*} - e^{2\kappa t_1} = \frac{(16F_0^2+4F_0-20)e^{2\kappa t_1} + 3(4F_0-1)}{18}.
\end{equation}
Let $e^{2\kappa t^*} - e^{2\kappa t_1} = 0$ we have
\begin{equation}
    t_1^* = \frac{1}{2\kappa}\ln\left[\frac{3(4F_0-1)}{20 - 4F_0 - 16F_0^2}\right],
\end{equation}
which has a positive solution for $\frac{3\sqrt{3}-2}{4}<F_0<1,~\mathrm{s.t.}~ \frac{4F_0-1}{20 - 4F_0 - 16F_0^2}>1$. Note that for $F_0=1$ the denominator $20 - 4F_0 - 16F_0^2=0$, while it is fine to consider $F_0<1$ as in practice there will be no perfect Bell pair. Moreover, the $t_1$ derivative of $e^{2\kappa t^*} - e^{2\kappa t_1}$ is
\begin{equation}
    \frac{d}{dt_1}\left[e^{2\kappa t^*} - e^{2\kappa t_1}\right](t_1) = \frac{8\kappa(4F_0^2+F_0-5)e^{2\kappa t_1}}{18},
\end{equation}
which will always be negative for $0<F_0<1$, i.e. the difference between $e^{2\kappa t^*}$ and $e^{2\kappa t_1}$ decreases as $t_1$ increases. We thus know that for $0<t_1<t_1^*$ we have $t^*>t_1$, and for $t_1>t_1^*$ we have $t^*<t_1$. 

In summary, to obtain the highest successful output fidelity at $t_2$, if the second entangled pair is successfully generated later than $t_1^*$, then the best choice is to perform the EPP immediately, while if the successful generation time of the second pair is earlier than $t_1^*$, the best choice is to perform the EPP at time $t^*$.
\end{proof}

\subsection{Proof of Proposition~\ref{thm:werner_normfid_opt_time}}
\begin{proof}
    We first examine the fidelity. The fidelity dynamics of the Werner state under a memory depolarizing channel is
    \begin{equation}
        F_\mathrm{Werner}(F_0,t) = F_0 e^{-2\kappa t} + \frac{1-e^{-2\kappa t}}{4}.
    \end{equation}
    Then the expression of normalized fidelity at $t_2$ is
    \begin{equation}
        \tilde{F}_\mathrm{succ}^{(W)}(t_2) = \frac{9 + (4F_0-1)\left[3\left(e^{-2\kappa(t_2-t_1)}+e^{-2\kappa t_2}\right) + (4F_0-1)\left(4e^{-2\kappa(t-t_1+t_2)} + e^{-2\kappa(2t-t_1)}\right)\right]}{72},
    \end{equation}
    where we use superscript (W) to denote the Werner state. The derivative w.r.t. $t$ is
    \begin{equation}\label{eqn:fid_norm_derivative_bbpssw}
        \frac{d}{dt}[\tilde{F}_\mathrm{succ}^{(W)}(t_2)](t) = -\frac{\kappa(4F_0-1)^2}{18}\left(2e^{-2\kappa(t-t_1+t_2)} + e^{-2\kappa(2t-t_1)}\right) < 0,
    \end{equation}
    which suggests the normalized fidelity at the end always decreases when purification is performed at later times.

    Then we consider the concurrence and the negativity naturally follows. The normalized concurrence at $t_2$ is 
    \begin{equation}
        \tilde{\mathcal{C}}_\mathrm{succ}^{(W)}(t_2) = \frac{3}{4} + \frac{e^{-2\kappa(t+t_2-t_1)}}{36}\left[\left(4+3e^{2\kappa(t_2-t)}\right)(4F_0-1) + 3e^{2\kappa t}\left(1+e^{-2\kappa t_1}\right)\right](4F_0-1),
    \end{equation}
    and the derivative w.r.t. $t$ is
    \begin{equation}
        \frac{d}{dt}[\tilde{\mathcal{C}}_\mathrm{succ}^{(W)}(t_2)](t) = -\frac{\kappa(4F_0-1)^2}{9}\left(2+3e^{2\kappa(t_2-t)}\right)e^{-2\kappa(t+t_2-t_1)} < 0.
    \end{equation}
    This $t$ derivative is always negative and thus this completes the proof for normalized concurrence.
\end{proof}

\section{Robustness of optimal time}\label{sec:robust_opt_time}
In this appendix we demonstrate the robustness of parameter-independent optimal time derived in the main text. We consider the optimal time for fidelity at $t_2$ conditioned on successful entanglement purification as an example. The key point is that, although in general cases the optimal time can be explicitly dependent on timing parameters such as $t_1,t_2$, and system parameters such as the raw entanglement fidelity of each Bell pair and the decoherence rates for different quantum memories, we can show that the optimal time is independent of these parameters in some practical regime. 

Specifically, we assume the following scenario: First a bit-flipped Bell pair with raw fidelity $F_1$ is generated at time 0, and at a later time $t_1\geq 0$ a second bit-flipped Bell pair with raw fidelity $F_2$ is generated. Between the time origin and time $t_1$ the first Bell pair decoheres. Again we only want to utilize \textit{one} Bell pair at a later time $t_2\geq t_1$. Without loss of generality, we measure the old Bell pair in EPP. We also assume that all 4 quantum memories undergo a bit-flip channel, and their decoherence rates are denoted by $\kappa_{1(2),1(2)}$, where the first subscript labels the Bell pair and the second subscript labels the quantum memory for storing one Bell pair. We consider the general case where $F_1$ can be different from $F_2$ and $\kappa_{1(2),1(2)}$ can be non-identical.
\begin{proposition}
    Under the assumed quantum network scenario, with bit-flipped Bell state under memory bit-flip channel, in order to obtain highest fidelity Bell state to utilize conditioned on successful purification, the recurrence EPP should always be performed at the latest possible time, if the utilization time $t_2$ satisfies:
    \begin{equation}
        t_2\leq\max\{(\kappa_{11}+\kappa_{12})^{-1},(\kappa_{21}+\kappa_{22})^{-1}\},
    \end{equation}
     and the quantum memory decoherence rates satisfy:
     \begin{equation}
         \frac{\kappa_{21}+\kappa_{22}}{\kappa_{11}+\kappa_{12}}\geq \frac{e^2 - (2F_2-1)^2}{e^2 + (2F_2-1)^2 + 2e^{-1}(2F_1-1)(2F_2-1)}.
     \end{equation}
\end{proposition}
\begin{proof}
    The entanglement fidelity at $t_2$ conditioned on successful entanglement purification at $t$ is
    \begin{equation}
        F_\mathrm{succ}(t_2) = \frac{\left[
        \begin{aligned}
            &e^{-(\kappa_{21}+\kappa_{22})t_1} + e^{-(\kappa_{11}+\kappa_{12})t-(\kappa_{21}+\kappa_{22})(t_1+t_2-t)}(2F_1-1) +\\
            &e^{-(\kappa_{21}+\kappa_{22})t_2}(2F_2-1) + e^{-(\kappa_{11}+\kappa_{12}+\kappa_{21}+\kappa_{22})t}(2F_1-1)(2F_2-1)
        \end{aligned}
        \right]}{2[e^{-(\kappa_{21}+\kappa_{22})t_1} + e^{-(\kappa_{11}+\kappa_{12}+\kappa_{21}+\kappa_{22})t}(2F_1-1)(2F_2-1)]}.
    \end{equation}
    Then we can also evaluate its derivative with respect to $t$, and discover that the denominator of the derivative is always non-negative. Therefore, we will focus on the time-dependent part of the derivative's numerator, where we have removed a common time-dependent factor which is always positive
    \begin{equation}
    \begin{aligned}
        &(\kappa_{21}+\kappa_{22})\left[(2F_2-1)^2 + e^{2(\kappa_{21}+\kappa_{22})(t-t_1)} + 2e^{-(\kappa_{11}+\kappa_{12})t+(\kappa_{21}+\kappa_{22})(t-t_1)}(2F_1-1)(2F_2-1)\right]\\
        &+ (\kappa_{11}+\kappa_{12})\left[(2F_2-1)^2 - e^{2(\kappa_{21}+\kappa_{22})(t-t_1)}\right],        
    \end{aligned}
    \end{equation}
    whose sign will determine the monotonicity of $F_\mathrm{succ}(t_2)$ with varying $t$. Now we start to include the practical assumption on utilization time $t_2\leq\max\{(\kappa_{11}+\kappa_{12})^{-1},(\kappa_{21}+\kappa_{22})^{-1}\}$, which is justified by the fact that in practice we will not store entangled states forever, but instead we will re-initialize quantum memories (cut-off) if they idle for too long time, as $(\kappa_{i1}+\kappa_{i2})^{-1}$ is the characteristic time scale of the fidelity decay for the $i$-th entangled pair. Now we use the assumption to upper bound the following quantity
    \begin{equation}
    \begin{aligned}
        &\frac{e^{2(\kappa_{21}+\kappa_{22})(t-t_1)} - (2F_2-1)^2}{e^{2(\kappa_{21}+\kappa_{22})(t-t_1)} + (2F_2-1)^2 + 2e^{-(\kappa_{11}+\kappa_{12})t+(\kappa_{21}+\kappa_{22})(t-t_1)}(2F_1-1)(2F_2-1)}\\
        \leq& \frac{e^{2(\kappa_{21}+\kappa_{22})t_2} - (2F_2-1)^2}{e^{2(\kappa_{21}+\kappa_{22})t_2} + (2F_2-1)^2 + 2e^{-(\kappa_{11}+\kappa_{12})t_2}(2F_1-1)(2F_2-1)}\\
        \leq& \frac{e^2 - (2F_2-1)^2}{e^2 + (2F_2-1)^2 + 2e^{-1}(2F_1-1)(2F_2-1)}.
    \end{aligned}
    \end{equation}
    It is then obvious that if
    \begin{equation}
        \frac{\kappa_{21}+\kappa_{22}}{\kappa_{11}+\kappa_{12}}\geq \frac{e^2 - (2F_2-1)^2}{e^2 + (2F_2-1)^2 + 2e^{-1}(2F_1-1)(2F_2-1)},
    \end{equation}
    the numerator of $t$ derivative of $F_\mathrm{succ}(t_2)$ will be positive, which means that the optimal time for performing the recurrence EPP is as late as possible.
\end{proof}
This proposition provides a sufficient condition for the optimal time to be still the latest possible time, when raw fidelities and memory decoherence rates can be different. We note that the upper bound on the RHS is never above 1 for all possible $F_1,F_2\in[1/2,1]$, because it can be shown that the highest value for it is taken at $F_1=F_2=1/2$, which gives value 1. Therefore, for reasonable values of raw fidelity strictly above 1/2, the sufficient condition suggests that the ratio $(\kappa_{21}+\kappa_{22})/(\kappa_{11}+\kappa_{12})$ can be both smaller than or larger than 1. This can be easily satisfied in reality if the decoherence rates of standardized quantum memories do not differ too much, which means that the parameter-independent optimal time is practically robust.

\section{Approximate location of optimal time transition border}\label{sec:approx_border}
As mentioned in the main text, the crossover between the two extreme strategies, i.e. the latest possible time and the earliest possible time, may be the result of sign change in the coefficient of $t$-linear part in the $t$ derivative of fidelity at $t_2$. Therefore, we are able to, at least approximately, determine the location of the border through the evaluation of the coefficient of the linear part in terms of relative strengths $a$ and $b$. We demonstrate this explicitly below.

In fact, we can obtain the expression of fidelity at $t_2$ conditioned on successful purification analytically
\begin{equation}
    F'(t) = \frac{\left[
    \begin{aligned}
    &e^{4(2\gamma - \gamma_x - \gamma_y)t_1 - 8(\gamma - \gamma_x - \gamma_y)t - 8\gamma t_2}\\
    \times&\left[9e^{8\gamma(t-t_1+t_2)} + 3(4F_0-1)e^{8\gamma t}\left(e^{-4(2\gamma - \gamma_x - \gamma_y)(t_1-t_2)} + e^{-8\gamma t_1 + 4(2\gamma - \gamma_x - \gamma_y)t_2}\right) \right.\\
    &+ (4F_0-1)^2\left(e^{4(\gamma+\gamma_x)(t-t_1+t_2)} + e^{4(\gamma+\gamma_x)t-4(\gamma+\gamma_y)(t_1-t_2)} + e^{4(\gamma+\gamma_x)(t-t_1)+4(\gamma+\gamma_y)t_2}\right.\\ 
    &+\left.\left.e^{4(\gamma - \gamma_x + 2\gamma_y)t + 4(\gamma+\gamma_x)t_2 - 4(\gamma+\gamma_y)t_1} + e^{8(\gamma - \gamma_x - \gamma_y)t + 8\gamma t_2 - 4(2\gamma - \gamma_x - \gamma_y)t_1} \right)\right]
    \end{aligned}\right]}{36e^{4(\gamma_x+\gamma_y)(2t-t_1)} + 4(4F_0-1)^2}.
\end{equation}
However, the complicated form forbids us from exactly solving for the value of $t$ such that $F'(t)=0$, which should in principle indicate the location where the optimal time transition happens.

Now we perform analytical approximation focusing on the leading order term of $t$. The first and most straightforward approach is to derive the linear term in the Taylor expansion of $F'(t)$ w.r.t. $t$. We can further simplify our treatment by performing linear approximation to the exponential functions involved in the expression for small exponents, i.e. $e^x\approx 1+x$ for $|x|\ll 1$. Thus in the second approximation approach, following the Taylor expansion of $F'(t)$, we can substitute the exponential functions in the linear term coefficient with their linear approximation. With one more step further, we can perform linear approximation of the exponential functions before the Taylor expansion, as the third approximation approach. It is clear that the level of approximation is increasing from the first approach to the third approach, and the accuracy is more and more dependent on the smallness of the exponents. It can be verified with specific choices of $(t_1, t_2, F_0)$ parameters that, while the estimated transition border location via the third approach can significantly deviate from the true result for larger $t_1$ and $t_2$, the third approach can obtain quite a good approximation for smaller $t_1$ and $t_2$. Therefore, the third approach, while bold, still captures the essence of the transition border location.

Now we derive analytical expression of the location of the border using linear term coefficient from the third approximation approach which gives the simplest form of the linear term coefficient among the three approaches
\begin{equation}
    c_t = \frac{\left[
    \begin{aligned}
    & 6(8F_0^2 + 1)[1 - 8t_2 + 4t_1(2 - w)](8F_0^2 - 4F_0 + 5 - 18t_1w)\\
    &+ 4[10 - w + 8F_0(2F_0 - 1)(1 - w) + 36(t_1 - 2t_2)w]\\
    &\times [6(t_1 - t_2) - 1 + 2(F_0 + (1-5F_0)(t_1 - 2t_2)w + F_0^2(4t_1(6 + w) - 8t_2(3 + w) - 5))]
    \end{aligned}
    \right]}{\left(8F_0^2 - 4F_0 + 5 - 18t_1w\right)^2},
    \label{eqn:lin_coeff_approx}
\end{equation}
where we have defined $w= x+y$. As this approximation is valid for small $t_1$ and $t_2$, which will thus be much smaller than $F_0$ and other $O(1)$ terms, we can perform further approximation
\begin{equation}
    c_t \approx \frac{6(8F_0^2 + 1)(8F_0^2 - 4F_0 + 5) - 4[10 - w + 8(2F_0^2 - F_0)(1 - w)](10F_0^2 - 2F_0 + 1)}{\left(8F_0^2 - 4F_0 + 5\right)^2},
\end{equation}    
from which we derive the root for $c_t=0$ as
\begin{equation}
    w(F_0) = x + y = \frac{8F_0^2 -4F_0 + 5}{20F_0^2 - 4F_0 + 2}.
    \label{eqn:anlytical_approx_border}
\end{equation}
The above equation defines a family of straight lines with slope $k=-1$, compatible with our observation from numerical results demonstrated in the main text. As an example, by substituting $F_0=0.95$ in the expression above we arrive at $w(0.95)\approx 0.518$ which is surprisingly close to 0.52 discovered numerically, even though this approximation of border location is obtained by approximations which are only valid for small $t_1$ and $t_2$.

\end{document}